\documentclass[runningheads,oribibl,envcountsame,envcountsect]{llncs}

\usepackage{graphicx}
\usepackage[utf8]{inputenc}

\usepackage{hyperref}

\urldef{\mailAndreNicole}\path|{frochaua,schweikn}@informatik.hu-berlin.de|

\usepackage{url}

\usepackage{amsmath}

\usepackage{amssymb}
\usepackage{amsfonts}
\usepackage{stmaryrd}
\usepackage{enumerate}
\usepackage{multicol}

\usepackage{tikz}
\usetikzlibrary{arrows}

\usetikzlibrary{arrows,positioning} 
\tikzset{
    >=stealth',
    punkt/.style={
           text width=13em,
           minimum height=2em,
           text centered},
    pil/.style={
           ->,
           thick,
           shorten <=2pt,
           shorten >=2pt,}
}

\newcounter{ruleBoxNo}

\usepackage{framed, color}
\definecolor{shadecolor}{gray}{.95}
\definecolor{framecolor}{gray}{.65}

{ \rightline{\textit{Box \theruleBoxNo}} \endMakeFramed}

  \newcommand{\EXPTIME}{\textnormal{\textnormal{\textsc{Exptime}}}}
  \newcommand{\TwoEXPTIME}{\textnormal{\textnormal{\textsc{2Exptime}}}}

          \newcommand{\T}{\ensuremath{\mathcal{T}}}

	\newcommand{\AF}[1]{{#1}}
	\newtheorem{Theorem}{Theorem}
    \newtheorem{Lemma}[Theorem]{Lemma}
    \newtheorem{Corollary}[Theorem]{Corollary}
     \DeclareMathOperator{\MSO}{\textnormal{MSO}}
     \DeclareMathOperator{\mDatalog}{\textnormal{mDatalog}}
          \DeclareMathOperator{\Root}{\textnormal{\textbf{root}}}
          \DeclareMathOperator{\Label}{\textnormal{\textbf{label}}}
           \DeclareMathOperator{\Child}{\textnormal{\textbf{child}}}
          \DeclareMathOperator{\Desc}{\textnormal{\textbf{desc}}}
          \DeclareMathOperator{\Fc}{\textnormal{\textbf{fc}}}
          \DeclareMathOperator{\Ns}{\textnormal{\textbf{ns}}}
          \DeclareMathOperator{\Ls}{\textnormal{\textbf{ls}}}
          \DeclareMathOperator{\Leaf}{\textnormal{\textbf{leaf}}}
          \DeclareMathOperator{\Reject}{\textnormal{\textbf{reject}}}
          \DeclareMathOperator{\idb}{\textnormal{idb}}
          \DeclareMathOperator{\edb}{\textnormal{edb}}
        \DeclareMathOperator{\state}{\textnormal{state}}
          
          \newcommand{\Yes}{\ensuremath{\textnormal{\textbf{yes}}}} 
          \newcommand{\No}{\ensuremath{\textnormal{\textbf{no}}}} 
         \newenvironment{proofof}[1]{ \bigskip \noindent\textbf{Proof
             of #1:}\\ }{\qed}
         \newcommand{\nc}{\newcommand}
          \nc{\rnc}{\renewcommand}
          \nc{\emptyquery}{\varnothing}
          \rnc{\geq}{\ensuremath{\geqslant}}
          \rnc{\leq}{\ensuremath{\leqslant}}
          \nc{\deff}{:=}
          \nc{\set}[1]{\ensuremath{\{ #1 \}}}
          \nc{\setc}[2]{\set{#1 \,:\, #2}}
         \nc{\NN}{\ensuremath{\mathbb{N}}}
          \nc{\NNpos}{\ensuremath{\NN_{\mbox{\tiny $\scriptscriptstyle \geq 1$}}}}
          \nc{\ar}{\ensuremath{\textit{ar}}} 
         \nc{\A}{\ensuremath{\mathcal{A}}}
          \nc{\PP}{\ensuremath{\mathcal{P}}}
          \rnc{\S}{\ensuremath{\mathcal{S}}}
          \nc{\tauGK}{\ensuremath{\tau_{\textit{GK}}}}
          \nc{\tauGKAlph}[1]{\ensuremath{\tau_{\textit{GK},#1}}}
          \nc{\tauGKSigma}{\tauGKAlph{\Sigma}}
          \nc{\tauo}{\ensuremath{\tau_o}}
          \nc{\tauu}{\ensuremath{\tau_u}}
          \nc{\tauuAlph}[1]{\ensuremath{\tau_{u,#1}}}
          \nc{\tauuSigma}{\tauuAlph{\Sigma}}
          \nc{\tauuSigmaStrich}{\tauuAlph{\Sigma'}}
          \nc{\atoms}{\ensuremath{\textit{atoms}}}
          \nc{\emptyword}{\varepsilon}
	\nc{\size}[1]{\ensuremath{|\!|#1|\!|}}
	\nc{\Bool}[1]{\ensuremath{#1_{\textit{Bool}}}}
	\nc{\Aut}[1]{\ensuremath{\texttt{\upshape #1}}} 
	\nc{\myparagraph}[1]{ \smallskip \textbf{#1.} }

\usetikzlibrary{matrix}
\usetikzlibrary{automata,positioning,shapes.arrows,chains}
\usetikzlibrary{trees,decorations.pathmorphing,decorations.markings} 
\usetikzlibrary{decorations,arrows}
\usetikzlibrary{decorations.pathmorphing}
\usepgflibrary{decorations.pathreplacing}
\usetikzlibrary{decorations.text}
\usepgflibrary{shapes.geometric}

\makeatletter
\newcommand{\pushright}[1]{\ifmeasuring@#1\else\omit\hfill$\displaystyle#1$\fi\ignorespaces}
\makeatother

      \DeclareMathOperator{\CQ}{\text{CQ}}
      \DeclareMathOperator{\strL}{\textnormal{strL}} 
\nc{\CT}{\textnormal{CT}}
\nc{\structure}{\textit{structure}}
\nc{\CTA}{\ensuremath{\CT_{\forall}}}
\nc{\CTE}{\ensuremath{\CT_{\exists}}}
\nc{\digit}{\ensuremath{\textit{digit}}}
\nc{\Ccount}{\ensuremath{\textit{count}}}
\nc{\Left}{\textnormal{left}}
\nc{\Right}{\textnormal{right}}
\nc{\CTEleft}{\CTE^{\textit{left}}}
\nc{\CTEright}{\CTE^{\textit{right}}}
\nc{\LeafCT}{\ensuremath{\textnormal{Leaf-}\CT}}
\nc{\LeafCTA}{\ensuremath{\textnormal{Leaf-}\CTA}}
\nc{\LeafCTEleft}{\textnormal{Leaf-}\CTEleft}
\nc{\LeafCTEright}{\textnormal{Leaf-}\CTEright}
\nc{\CCells}{\textnormal{CCells}}
\nc{\BCells}{\textnormal{BCells}}
\nc{\PCells}{\textnormal{PCells}}
\nc{\Non}{\textnormal{Non}}
\nc{\Blank}{\textbf{\textvisiblespace}}
\nc{\startCT}{\ensuremath{\textnormal{Start-}\CT}}
 \DeclareMathOperator{\enc}{\textnormal{enc}}
 \nc{\Succ}{\textnormal{Succ}}
  \nc{\SameLevel}{\textit{SameLevel}}
   \nc{\SameCell}{\textit{SameCell}}
  \nc{\EquiLevel}{\textit{EquiLevel}}
   \nc{\EquiCell}{\textit{EquiCell}}
      \nc{\Acc}{\ensuremath{\textit{Acc}}}
                 \nc{\Ans}{\ensuremath{\textit{Ans}}}
 \nc{\type}{\textit{type}}

\let\doendproof\endproof \renewcommand\endproof{~\hfill$\qed$\doendproof}

\begin{document}

\mainmatter 
\title{Monadic Datalog Containment on Trees Using the Descendant-Axis%
  \protect\footnote{This article is the full version of  \cite{FrochauxS16}.}%
}
\titlerunning{%
  Monadic Datalog Containment on Trees Using the Descendant-Axis%
}%
\author{Andr\'e Frochaux \and Nicole Schweikardt}
\institute{
Institut für Informatik, Humboldt-Universität zu Berlin\\ \mailAndreNicole
}

\maketitle

\begin{abstract} In their AMW'14-paper,
Frochaux, Grohe, and Schweikardt showed that the query containment problem for monadic datalog on
finite unranked labeled trees is $\EXPTIME$-complete when
(a) considering unordered trees using the \emph{child}-axis,
and when (b) considering ordered trees using the axes
\emph{firstchild}, \emph{nextsibling}, and \emph{child}.
Furthermore, when allowing to use also the \emph{descendant}-axis, the query containment problem was shown to be solvable in 2-fold exponential time,
but it remained open to determine the problem's exact complexity in presence of the descendant-axis.
The present paper closes this gap by showing that, in the presence of the descendant-axis, the problem is $\TwoEXPTIME$-hard. 
\end{abstract}

\setlength\abovedisplayshortskip{0pt}
\setlength\belowdisplayshortskip{0pt}	
\setlength\abovedisplayskip{4pt plus 2pt minus 2pt}    
\setlength\belowdisplayskip{4pt plus 2pt minus 2pt}

\section{Introduction}\label{section:Introduction}

The query containment problem (QCP) is a fundamental problem that has been
studied for various query languages. Datalog is a standard tool for
expressing queries with recursion. From Cosmadakis et al.\ \cite{CGKV}
and Benedikt et al.\ \cite{DBLP:conf/icalp/BenediktBS12} it is known
that the QCP for \emph{monadic} datalog queries on the class of
all finite relational structures is $\TwoEXPTIME$-complete.

Restricting attention to finite unranked labeled trees,
Gottlob and Koch \cite{GottlobKoch}
showed that on \emph{ordered} trees the QCP for monadic datalog is 
$\EXPTIME$-hard and decidable, leaving open the question for a tight bound.
This gap was closed by Frochaux, Grohe, and Schweikardt
in \cite{FroGroSchw14} by giving a matching $\EXPTIME$ upper bound for
the QCP for monadic datalog on ordered trees using the axes
\emph{firstchild}, \emph{nextsibling}, and \emph{child}.
Similar results were obtained in \cite{FroGroSchw14} also for 
\emph{unordered} finite labeled trees: in this setting, the QCP is
$\EXPTIME$-complete for monadic datalog queries on unordered trees
using the \emph{child}-axis.

For the case where queries are allowed to also use the 
\emph{descendant}-axis, \cite{FroGroSchw14} presented a 2-fold
exponential time algorithm for the QCP for monadic datalog on (ordered
or unordered) trees. Determining the problem's exact complexity in the
presence of the \emph{descendant}-axis, however, was left open.

The present paper closes the gap by proving a matching $\TwoEXPTIME$
lower bound (both, for ordered and for unordered trees).
This gives a conclusive answer to a question posed by Abiteboul et
al.\ in \cite{AbiteboulBMW13}, asking for the complexity of the QCP
on unordered trees in the presence of the  descendant-axis.
Our $\TwoEXPTIME$-hardness proof for \emph{ordered} trees is by a reduction
from a $\TwoEXPTIME$-hardness result of \cite{BjorklundMS08} for the
validity of conjunctive queries w.r.t.\ schema constraints.
For obtaining the $\TwoEXPTIME$-hardness on \emph{unordered} trees, we
follow the approach of \cite{BjorklundMS08} and
construct a reduction from the $\TwoEXPTIME$-complete word problem for
exponential-space bounded alternating Turing machines \cite{DBLP:journals/jacm/ChandraKS81}.

The remainder of the paper is organised as follows. 
Section~\ref{section:Preliminaries} fixes the basic notation.
Section~\ref{section:ReductionNTA} presents a $\TwoEXPTIME$ lower
bound for the QCP  on ordered trees using the axes
\emph{firstchild}, \emph{nextsibling}, \emph{root}, \emph{leaf},
\emph{lastsibling}, \emph{child}, \emph{descendant}.
Section~\ref{section:Emptiness} is devoted to the $\TwoEXPTIME$ lower bound
for the QCP on unordered trees using only the axes \emph{child} and \emph{descendant}.
We conclude in Section~\ref{section:relatedResults}. 

  Proof details can be found in the appendix.

\section{Trees and Monadic Datalog ($\mDatalog$)}\label{section:Preliminaries}\label{section:Datalog}

Throughout this paper, $\Sigma$ will
always denote a finite non-empty alphabet.
\\
By $\NN$ we denote the set of non-negative integers, and we let $\NNpos\deff\NN\setminus\set{0}$.

\myparagraph{Relational Structures}
As usual, a \emph{schema} $\tau$ consists of a
finite number of relation symbols $R$, each of a fixed
\emph{arity} $\ar(R)\in\NNpos$.
A \emph{$\tau$-structure} $\A$ consists of a \emph{finite} non-empty
set $A$ called the \emph{domain}
of $\A$, and a relation
$R^{\A}\subseteq A^{\ar(R)}$ for each relation symbol $R\in\tau$. 
It will often be convenient to
identify $\A$ with the \emph{set of atomic facts of $\A$},
i.e., the set $\atoms(\A)$ consisting of all facts $R(a_1,\ldots,a_{\ar(R)})$
for all relation symbols $R\in\tau$ and all tuples $(a_1,\ldots,a_{\ar(R)})\in R^\A$.

If $\tau$ is a schema and $\ell$ is a list of relation symbols, we
write $\tau^\ell$ to denote the extension of the schema $\tau$ by the
symbols in $\ell$.
Furthermore, $\tau_\Sigma$ denotes the extension of $\tau$ by
new unary relation symbols $\Label_\alpha$, for all $\alpha\in\Sigma$.

\myparagraph{Unordered Trees}
An \emph{unordered $\Sigma$-labeled tree} $T=(V^T,\lambda^T,E^T)$
consists of a finite non-empty set $V^T$ of nodes, a function  
$\lambda^T: V^T\to \Sigma$ assigning to each node $v$ of $T$ a label
$\lambda(v)\in\Sigma$, and a set $E^T\subseteq V^T\times V^T$ of directed edges such
that the directed graph $(V^T,E^T)$ is a rooted tree where edges are directed
from the root towards the leaves.
We represent such a tree $T$ as a relational structure of domain $V^T$
with unary and binary relations: For each label $\alpha\in
\Sigma$, $\Label_\alpha(x)$ expresses that $x$ is a node with label
$\alpha$;
$\Child(x,y)$ expresses that $y$ is a child of node $x$; 
$\Root(x)$ expresses that $x$ is the tree's root node;
$\Leaf(x)$ expresses that $x$ is a leaf; and $\Desc(x,y)$ expresses
that $y$ is a descendant of $x$ (i.e., $y$ is a child or a grandchild
or \ldots\ of $x$).
We denote this relational structure
representing $T$ by $\S_u(T)$, but when no confusion
arises we simply write $T$ instead of $\S_u(T)$.

The queries we consider for unordered trees are allowed to make use of at least the
predicates $\Label_\alpha$ and $\Child$.
We fix the schema
\ \(
  \tauu \ \deff \ \set{\,\Child\,}.
\) \

\myparagraph{Ordered Trees}
An \emph{ordered} $\Sigma$-labeled tree $T=(V^T,\lambda^T,E^T,\textit{order}^T)$
has the same components as an unordered $\Sigma$-labeled tree and,
in addition, $\textit{order}^T$ fixes for each node $u$ of $T$
a strict linear order of all the children of $u$ in $T$.

To represent such a tree as a relational structure, we use the same
domain and the same 
predicates as for unordered $\Sigma$-labeled trees, along with three further
predicates $\Fc$ (``first-child''), $\Ns$ (``next-sibling''), and
$\Ls$ (``last sibling''), where $\Fc(x,y)$ expresses that $y$
is the first child of node $x$ (w.r.t.\ the linear order of the
children of $x$ induced by $\textit{order}^T$); $\Ns(x,y)$
expresses that $y$ is the right sibling of $x$ (i.e., $x$ and $y$ have
the same parent $p$, and $y$ is the immediate successor of $x$ in the
linear order of $p$'s children given by $\textit{order}^T$); and 
$\Ls(x)$ expresses that $x$ is the rightmost sibling (w.r.t.\ the
linear order of the children of $x$'s parent given by $\textit{order}^T$).
We denote this relational structure representing $T$ by $\S_o(T)$,
but when no confusion arises we simply write $T$ instead of $\S_o(T)$.

The queries we consider for ordered trees are allowed to make use of
at least the predicates $\Label_\alpha$, $\Fc$, and $\Ns$.
We fix the schemas
\ $\tauo \deff  \set{\,\Fc, \, \Ns\,}$ \ and 
\ $\tauGK\ \deff \ \tauo^{\Root,\Leaf,\Ls}$. \
In \cite{GottlobKoch}, Gottlob and Koch used $\tauGKSigma$-structures
to represent
ordered $\Sigma$-labeled trees.

\myparagraph{Datalog}
We assume that the reader is familiar with the syntax and semantics of
\emph{datalog} (cf., e.g., \cite{Datalog-Bible,GottlobKoch}).
Predicates that occur in the head of
some rule of 
a datalog program
$\PP$ are called \emph{intensional}, whereas predicates that only
occur in the body of rules of $\PP$ are called \emph{extensional}.
By $\idb(\PP)$ and
$\edb(\PP)$ we denote the sets of intensional and extensional predicates
of $\PP$, resp. We say that $\PP$ \emph{is of schema $\tau$}
if $\edb(\PP)\subseteq \tau$.
We write $\T_{\PP}$ to denote the \emph{immediate
  consequence operator} associated with a datalog program $\PP$.
Recall that $\T_{\PP}$ maps a set $C$ of atomic facts to the set of all
atomic facts that are derivable from $C$ by at most one
application of the rules of $\PP$.
The monotonicity of $\T_{\PP}$
implies that for each finite set $C$, the iterated application of
$\T_{\PP}$ to $C$ leads to a fixed point, denoted by
$\T_{\PP}^\omega(C)$, which is reached after a finite number of
iterations.

\myparagraph{Monadic datalog queries}
A datalog program belongs to \emph{monadic datalog}
($\mDatalog$, for short), 
if all its \emph{intensional} predicates have arity~1. 

A \emph{unary monadic datalog query} of schema $\tau$ is a tuple $Q=(\PP,P)$
where $\PP$ is a monadic datalog program of schema $\tau$ and $P$ is an
intensional predicate of $\PP$.
$\PP$ and $P$ are called the \emph{program} and the \emph{query
  predicate} of $Q$.
When evaluated in a finite $\tau$-structure $\A$ that represents a
labeled tree
$T$, the query $Q$ results in the unary relation
\,$\AF{Q}(T) \, \deff \,
  \setc{ a\in A\,}{\,  P(a) \,\in\,
    \T_\PP^\omega(\atoms(\A))\, }$.

The \emph{Boolean monadic datalog query} $\Bool{Q}$ specified by
$Q=(\PP,P)$ is the Boolean query with $\AF{\Bool{Q}}(T)=\Yes$ iff the
tree's root node belongs to $\AF{Q}(T)$.

The \emph{size} $\size{Q}$ of a monadic datalog query $Q$ is the
length of $Q=(\PP,P)$ viewed as a string over a suitable alphabet.

\myparagraph{Expressive power of monadic datalog on trees} 
On \emph{ordered}
$\Sigma$-labeled trees represented as $\tauGKSigma$-structures, 
monadic datalog can express exactly the same unary queries as monadic
second-order logic \cite{GottlobKoch} --- for short, we will say ``$\mDatalog(\tauGK)=\MSO(\tauGK)$ on
ordered trees''. Since the $\Child$ and $\Desc$ relations are
definable in $\MSO(\tauGK)$, \
$\mDatalog(\tauGK)=\mDatalog(\tauGK^{\Child,\Desc})$ on ordered trees.
Moreover, for (ordered or unordered) trees, every monadic Datalog query that uses the $\Desc$-axis
can be rewritten in 1-fold exponential time into an equivalent monadic
datalog query
which uses the $\Child$-axis, but not the $\Desc$-axis (see the proof
of Lemma~23 in the full version of \cite{FroGroSchw14}).

Using the monotonicity of the immediate
consequence operator, 
one obtains
that removing any of the predicates $\Root,
\Leaf,\Ls$ from $\tauGK$ strictly decreases the expressive power of
$\mDatalog$ on ordered trees (see \cite{DBLP:journals/corr/FroSchw13}). 
By a similar reasoning
one also obtains that on \emph{unordered} trees,
represented as $\tauuSigma^{\Root,\Leaf,\Desc}$-structures, monadic
datalog is strictly less expressive than monadic
second-order logic, and omitting any of the predicates $\Root$, $\Leaf$ 
further reduces the expressiveness of
monadic datalog on unordered trees \cite{DBLP:journals/corr/FroSchw13}.

\myparagraph{The Query Containment Problem (QCP)} Let $\tau_\Sigma$ be one of the schemas used for representing (ordered or unordered) 
$\Sigma$-labeled trees as relational structures. For two unary queries $Q_1$ and $Q_2$ of schema $\tau_\Sigma$
we write $Q_1 \subseteq Q_2$ 
to indicate that for every $\Sigma$-labeled tree $T$ we have $\AF{Q_1}(T) \subseteq \AF{Q_2}(T)$.
Similarly, if $Q_1$ and $Q_2$ are \emph{Boolean} queries of schema $\tau_\Sigma$, we write
$Q_1\subseteq Q_2$ to indicate that for every $\Sigma$-labeled tree
$T$, if $\AF{Q_1}(T)=\Yes$ then also $\AF{Q_2}(T)=\Yes$.
We write $Q_1\not\subseteq Q_2$ to indicate that
$Q_1\subseteq Q_2$ does not hold. 
For a schema $\tau$, the \emph{query containment problem (QCP) for $\mDatalog(\tau)$ on finite labeled trees} receives as input a
finite alphabet $\Sigma$ and two (unary or Boolean)
$\mDatalog(\tau_\Sigma)$-queries $Q_1$ and $Q_2$, and the task is to
decide whether $Q_1\subseteq Q_2$.
From \cite{FroGroSchw14} we know:
\begin{Theorem}[Frochaux et al.\ \cite{FroGroSchw14}]\label{thm:AMW14}
 The QCP for $\mDatalog(\tau_u^{\Root,\Leaf,\Desc})$ on unordered trees and
 for $\mDatalog(\tauGK^{\Child,\Desc})$ on ordered trees can be solved in 2-fold exponential time.
\end{Theorem}

\section{$\TwoEXPTIME$-hardness on Ordered Trees}
\label{section:ReductionNTA}

\newcounter{CounterThm:2EXPhard_ordered_unranked}
\setcounter{CounterThm:2EXPhard_ordered_unranked}{\value{theorem}}

\begin{Theorem}\label{Thm:2EXPhard_ordered_unranked}
  The QCP for Boolean $\mDatalog(\tauGK^{\Child,\Desc})$ on
  finite labeled ordered trees is $\TwoEXPTIME$-hard.
\end{Theorem}
The proof is by a reduction based on a $\TwoEXPTIME$-hardness result of Björklund, Martens, and Schwentick \cite{BjorklundMS08}.
For stating their result, we recall some notation used in \cite{BjorklundMS08}. 
A \emph{nondeterministic (unranked) tree automaton} (NTA) $\Aut{A}=(\Sigma,S,\Delta,F)$ consists of an input alphabet $\Sigma$, 
a finite set $S$ of states, a set $F \subseteq S$ of accepting states, and a finite set $\Delta$ of transition rules of the form 
$(s,\alpha) \rightarrow L$, where $s\in S$, $\alpha \in \Sigma$, and
$L$ is a regular string-language over $S$. 
A \emph{run} of the NTA $\Aut{A}$ on a ordered $\Sigma$-labeled tree $T$ is a mapping $\rho: V^T \rightarrow S$ such that the following 
is true for all nodes $v$ of $T$, where $\alpha$ denotes the label of $v$ in $T$:
 if $v$ has $n\geq 0$ children $u_1,\ldots,u_n$ (in order from the left to the right), then there exists a rule 
     $(s,\alpha) \rightarrow L$ in $\Delta$ such that $\rho(v)=s$ and
     $w_v  \in L$, for the string $w_v\deff \rho(u_1)\cdots
     \rho(u_n)$. In particular, if $v$ is a leaf, then there must be a rule $(s,\alpha) \rightarrow L$ in $\Delta$ such that  $\rho(v)=s$ and $\varepsilon \in L$, where $\varepsilon$ denotes the empty string. 

A run $\rho$ of $\Aut{A}$ on $T$ is \emph{accepting}, if $T$'s root note $v$ is labeled with an accepting state of $\Aut{A}$, i.e., $\rho(v)\in F$.
A finite ordered $\Sigma$-labeled tree $T$ is \emph{accepted} by $\Aut{A}$, if there exists an accepting run of $\Aut{A}$ on $T$.
We write $L(\Aut{A})$ to denote the \emph{language} of $\Aut{A}$, i.e., the set of all finite ordered $\Sigma$-labeled trees that are accepted by $\Aut{A}$. 

To present an NTA $\Aut{A}=(\Sigma, S, \Delta, F)$ as an input for an algorithm, the string-languages $L$ that occur in the right-hand side of 
rules in $\Delta$ are specified by NFAs
$\Aut{A}_L=(\Sigma_L,Q_L,\delta_L,q_L,F_L)$, whose input alphabet
is $\Sigma_L \deff S$,
and where $Q_L$ is a finite set of states, $\delta_L \subseteq  (Q_L \times \Sigma_L \times Q_L)$ is a transition relation, 
$q_L \in Q_L$ is the initial state, and  $F_L \subseteq Q_L$ is the set of accepting states of $\Aut{A}_L$.
The \emph{size of $\Aut{A}_L$} is \ $\size{\Aut{A}_L}\deff
|Q_L|+|\delta_L|$, \ and the \emph{size of $\Aut{A}$} is the sum of
$|\Sigma|$, $|S|$, $|\Delta|$, and $\size{\Aut{A}_L}$, for all
$L\in\strL(\Aut{A})$, where 
$\strL(\Aut{A})$ is the set of all string-languages $L$ that occur in
the right-hand side of a rule in $\Delta$. 

In \cite{BjorklundMS08}, NTAs are used to describe schema information.
A Boolean query $Q$ is said to be \emph{valid with respect to} an NTA $\Aut{A}$
if $Q(T)=\Yes$ for every ordered $\Sigma$-labeled tree $T\in L(\Aut{A})$. 
The particular queries of interest here are 
\emph{Boolean $\CQ(\Child,\Desc)$} queries, i.e., Boolean conjunctive queries of 
schema $\tau_{u,\Sigma}^{\Desc}=\set{\Child,\Desc}\cup\setc{\Label_\alpha}{\alpha\in\Sigma}$, for a suitable alphabet $\Sigma$.
The problem \emph{``validity of Boolean $\CQ(\Child,\Desc)$ w.r.t. a tree automaton''} receives as input a Boolean $\CQ(\Child,\Desc)$ query $Q$ and
an NTA $\Aut{A}$, and the task is to decide whether $Q$ is valid with respect to $\Aut{A}$.

\begin{Theorem}[Björklund et al. \cite{BjorklundMS08}] \label{Thm:Validit}
Validity of Boolean $\CQ(\Child,\Desc)$ \allowbreak w.r.t.\ a tree automaton is $\TwoEXPTIME$-complete.
\end{Theorem}

Our proof of Theorem~\ref{Thm:2EXPhard_ordered_unranked} is via a polynomial-time reduction from 
the problem \emph{validity of Boolean $\CQ(\Child,\Desc)$ w.r.t.\ a
tree automaton} to the \emph{QCP for Boolean $\mDatalog(\tauGK^{\Child,\Desc})$ on finite labeled ordered trees}.

Let $Q_{\CQ}$ be a Boolean $\CQ(\Child,\Desc)$-query, and let $\Aut{A}$ be an NTA with input alphabet $\Sigma$. We translate $Q_{\CQ}$ into 
an equivalent $\mDatalog(\tau_{u,\Sigma}^{\Desc})$-query $Q'_{\CQ}=(\PP,P)$: \ If $Q_{\CQ}$ is of the form 
\ $\Ans()\leftarrow R_1(u_1),\ldots,R_\ell(u_\ell)$ \ for relational atoms $R_1(u_1),\ldots,R_\ell(u_\ell)$, we choose an arbitrary variable $x$ that 
occurs in at least one of these atoms, we use a new unary idb-predicate $P$, and we let $\PP$ be the program consisting of the two rules
\ $P(x) \leftarrow R_1(u_1),\ldots,R_\ell(u_\ell)$ \ and 
\ $P(x) \leftarrow \Child(x,y),P(y)$.

Then, for every ordered $\Sigma$-labeled tree $T$ we have \,$Q'_{\CQ,\textit{Bool}}(T)=\Yes$\, iff \,$Q_{\CQ}(T)=\Yes$.
The following Lemma~\ref{Lem:NTA2MDATALOG} constructs, in time polynomial in the size of $\Aut{A}$, an
$\mDatalog(\tauGKSigma^{\Child})$-query $Q_{\Aut{A}}$ which is equivalent to $\Aut{A}$, i.e., for every ordered $\Sigma$-labeled tree $T$ we have
\,$Q_{\Aut{A},\textit{Bool}}(T)=\Yes$\, iff \,$T\in L(\Aut{A})$.

Note that \ $Q_{\CQ}$ is valid w.r.t.\ $\Aut{A}$ \ if, and only if, \ $Q_{\Aut{A},\textit{Bool}}\subseteq Q'_{\CQ,\textit{Bool}}$.
Thus, we obtain the desired polynomial-time reduction, showing that the QCP for Boolean $\mDatalog(\tauGK^{\Child,\Desc})$ on finite ordered 
$\Sigma$-labeled trees inherits the $\TwoEXPTIME$-hardness from the problem ``validity of Boolean $\CQ(\Child,\Desc)$ w.r.t.\ a
tree automaton''.
All that remains to finish the proof of Theorem~\ref{Thm:2EXPhard_ordered_unranked} is to prove 
the following Lemma~\ref{Lem:NTA2MDATALOG}.

\begin{Lemma} \label{Lem:NTA2MDATALOG}
For every NTA $\Aut{A}=(\Sigma,S,\Delta,F)$  there is an  $\mDatalog(\tauGKSigma^{\Child})$-query $Q=(\PP,P)$, 
such that for every finite ordered $\Sigma$-labeled tree $T$ we have
\,$\AF{\Bool{Q}}(T)=\Yes$\, iff \,$T\in L(\Aut{A})$.
Furthermore, $Q$ is constructible from $\Aut{A}$ in time polynomial in the size of $\Aut{A}$.
\end{Lemma}

\begin{proof}
We construct a monadic datalog program $\PP$ which, for every node $v$ of $T$, computes information on \emph{all} states that $\Aut{A}$ can assume at node $v$,
i.e., all states $s\in S$ for which
there is a run $\rho$ of $\Aut{A}$ on the subtree of $T$ rooted at $v$, such that $\rho(v)=s$.
To this end, for every state $s\in S$, we will use an idb-predicate $s$. The query $Q_{\textit{Bool}}$ will accept an input tree $T$ if there is an
accepting state $s\in F$ such that
$s(\textit{root}^T) \in \T^{\omega}_{\PP}(T)$, where $\textit{root}^T$ denotes the root of $T$.
The program $\PP$ is constructed in such a way that it performs a generalised version of the well-known \emph{powerset construction}.

Recall that the transition rules of $\Aut{A}$ are of the form $(s,\alpha)\to L$, where $s\in S$, $\alpha\in\Sigma$, and $L$ is a regular string-language
over $S$, specified by an NFA $\Aut{A}_L=(\Sigma_L,Q_L,\delta_L,q_L,F_L)$ with $\Sigma_L=S$ and $\delta_L\subseteq (Q_L\times \Sigma_L\times Q_L)$.
W.l.o.g., we assume that the state sets of all the NFAs are mutually disjoint, and disjoint with $S$.

To emulate the standard powerset construction of the NFA $\Aut{A}_L$, we use an idb-predicate $q$ for every state $q\in Q_L$, and an
extra idb-predicate $\Acc_L$. If $u_1,\ldots,u_n$ are the children of a node $v$
in an input tree $T$, the NFA $\Aut{A}_L$ processes the strings over alphabet $S$ that are of the form $s_1\cdots s_n$, where
$s_i$ is a state that $\Aut{A}$ can assume at node $u_i$ (for every $i\in\set{1,\ldots,n}$).
We start by letting $\PP_L\deff \emptyset$ and then add to $\PP_L$ the following rules:
For the initial state $q_L$ of $\Aut{A}_L$, consider all $s\in S$ and $q\in Q_L$ such that   
$(q_L,s, q)\in \delta_L$, and add to $\PP_L$ the rule 
\[ q(x) \ \leftarrow \ \Fc(y,x),\ s(x)\,.  \]
Afterwards, for every transition $(q,s, q')\in \delta_L$, add to $\PP_L$ the rule
\[ q'(x') \ \leftarrow \  q(x),\ \Ns(x,x'),\ s(x')\,.  \]
Finally, for every accepting state $q\in F_L$ of $\Aut{A}_L$, add to $\PP_L$ the rule
\[\Acc_L(x) \ \leftarrow \ \Ls(x),\ q(x)\,.\]
Clearly, the program $\PP_L$ can be constructed in time polynomial in $\size{\Aut{A}_L}$.

Now, we are ready to construct the monadic datalog program $\PP$ that simulates the NTA $\Aut{A}$. 
We start by letting $\PP$ be the disjoint union of the programs $\PP_L$, for all $L\in\strL(\Aut{A})$.
The computation of $\Aut{A}$ on an input tree $T$ starts in the leaves of $T$. Thus, to initiate the simulation of 
$\Aut{A}$, we consider every rule $(s,\alpha) \rightarrow L$ in $\Delta$, where $\varepsilon \in L$.\footnote{Note that ``$\emptyword\in L$\,?'' can be checked 
by simply checking whether $q_L\in F_L$.}
For each such rule, we add to $\PP$ the rule 
\[ s(x) \ \leftarrow \  \Label_{\alpha}(x),\ \Leaf(x)\,.\]
Note that for each $L\in\strL(\Aut{A})$, the program $\PP_L$ ensures that every \emph{last}~sibling $u_n$ of a node $v$ will be marked by $\Acc_L(u_n)$ 
iff the states of $\Aut{A}$ assigned to $u_n$ and its siblings form a string in $L$. 
To transfer this information from the last sibling to its parent node, we add to $\PP$ the rule
\[ \Child_{\Acc_L}(y) \ \leftarrow \ \Child(y,x),\ \Ls(x),\ \Acc_L(x)\,, \]
where $\Child_{\Acc_L}$ is a new idb-predicate, for every $L\in\strL(\Aut{A})$.
\\
Afterwards, we consider every rule \ $(s,\alpha) \rightarrow L$ \ in $\Delta$, and add to $\PP$ the rule
\[  s(x) \ \leftarrow \ \Child_{\Acc_L}(x),\ \Label_{\alpha}(x)\,. \]
Finally, to test if $\Aut{A}$ accepts an input tree $T$, we add rules to test whether $T$'s root is assigned an accepting state of $\Aut{A}$. 
To this end, we consider every accepting state $s\in F$ of $\Aut{A}$ and add to $\PP$ the rule
\[ P(x) \ \leftarrow \ \Root(x),\  s(x)\,. \]
This finishes the construction of the program $\PP$ and the query $Q=(\PP,P)$. Clearly, $\PP$ is a monadic datalog program of schema
$\tauGKSigma^{\Child}$, and $Q$ can be constructed in time polynomial in $\size{A}$.
It is not difficult, but somewhat tedious, to verify that, as intended by the construction, indeed 
for every finite ordered $\Sigma$-labeled tree $T$ we have $Q_{\textit{Bool}}(T)=\Yes$ if, and only if, there exists an accepting run of
the NTA $\Aut{A}$ on $T$.
This completes the proof of Lemma~\ref{Lem:NTA2MDATALOG}.
\end{proof}

\section{$\TwoEXPTIME$-hardness on Unordered Trees }\label{section:Emptiness}

Our next aim is to transfer the statement of
Theorem~\ref{Thm:2EXPhard_ordered_unranked} to \emph{unordered} trees.
Precisely, we will show the following.

\newcounter{CounterThm:Desc_QCP_unranked_unordered}
\setcounter{CounterThm:Desc_QCP_unranked_unordered}{\value{theorem}}

\begin{Theorem}\label{Thm:Desc_QCP_unranked_unordered}
The QCP for Boolean $\mDatalog(\tauu^{\Desc})$ on finite labeled unordered trees is $\TwoEXPTIME$-hard. 
\end{Theorem}

For proving Theorem~\ref{Thm:Desc_QCP_unranked_unordered}, we cannot directly build on Björklund et al.'s
Theorem~\ref{Thm:Validit}, since their NTAs explicitly refer to
\emph{ordered} trees.

By constructing suitable reductions, we can show that proving
Theorem~\ref{Thm:Desc_QCP_unranked_unordered} 
boils down to proving the following
Theorem~\ref{Thm:Desc_Empti_ranked_ordered}, which deals with the
\emph{emptiness problem} on trees over a \emph{ranked} alphabet.

For the remainder of this section, $\Sigma'$ will denote a \emph{ranked}
finite alphabet. I.e., $\Sigma'$ is a finite set of symbols, and
each symbol $\alpha\in\Sigma'$ is equipped with a fixed \emph{arity}
$\ar(\alpha)\in\NN$. An unordered \emph{ranked} $\Sigma'$-labeled tree
is an unordered $\Sigma'$-labeled tree where each node labeled
with symbol $\alpha\in\Sigma'$ has exactly $\ar(\alpha)$ children.
For a Boolean $\mDatalog(\tauuSigmaStrich^{\Desc})$-query $Q$, we say that 
$Q$ is \emph{unsatisfiable by unordered ranked trees} (in symbols:
$Q=\emptyquery$) if for every
finite unordered ranked $\Sigma'$-labeled tree $T$ we have $Q(T)=\emptyset$. 
The \emph{emptiness problem for Boolean $\mDatalog(\tauuSigmaStrich^{\Desc})$ on
finite unordered ranked $\Sigma'$-labeled trees} receives as input a
Boolean $\mDatalog(\tauuSigmaStrich^{\Desc})$-query $Q$, and the task is to
decide whether $Q=\emptyquery$. 
The main technical step needed for proving
Theorem~\ref{Thm:Desc_QCP_unranked_unordered} is to prove the following.

\newcounter{CounterThm:Desc_Empti_ranked_ordered}
\setcounter{CounterThm:Desc_Empti_ranked_ordered}{\value{theorem}}
 
 \begin{Theorem} \label{Thm:Desc_Empti_ranked_ordered}
 There is a ranked finite alphabet $\Sigma'$, such that 
 the emptiness problem for Boolean $\mDatalog(\tauuSigmaStrich^{\Desc})$ on
 finite unordered ranked $\Sigma'$-labeled trees is \TwoEXPTIME-hard.
 \end{Theorem}

For the proof of Theorem~\ref{Thm:Desc_Empti_ranked_ordered}, we can
build on the approach 
used by Björklund et al.\ for proving 
Theorem~\ref{Thm:Validit}:
As in \cite{BjorklundMS08}, we proceed by a reduction from the word problem for 
exponential-space bounded alternating Turing machines, which is known
to be $\TwoEXPTIME$-complete
\cite{DBLP:journals/jacm/ChandraKS81}. The remainder of this section
is devoted to the proof of Theorem~\ref{Thm:Desc_Empti_ranked_ordered}.

An \emph{alternating Turing machine} (ATM) is a nondeterministic Turing machine
$\Aut{A}=(Q,\Sigma,\Gamma,\delta,q_0)$ 
whose 
state space $Q$ is partitioned into \emph{universal states} $Q_{\forall}$,
\emph{existential states} $Q_{\exists}$, an accepting state $q_a$, and
a \emph{rejecting state} $q_r$.
The ATM's tape cells are numbered 0,1,2,\ldots.
A \emph{configuration} of $\Aut{A}$ is a finite string of
the form $w_1qw_2$ with $w_1, w_2 \in \Gamma^*$ and $q \in Q$,
representing the situation where the ATM's tape contains the word $w_1w_2$,
followed by blanks, the ATM's current state is $q$, and the head is positioned at the first letter of
$w_2$.
A configuration $w_1qw_2$ is a \emph{halting} (\emph{universal}, \emph{existential},
resp.)
configuration if $q\in\set{q_a,q_r}$ ($q\in Q_\forall$, $q\in
Q_\exists$, resp.).
W.l.o.g., no halting configuration has a successor configuration, and every halting
configuration is of the form $qw$.  
A
 \emph{computation tree} $T_{\Aut{A}}$ of the ATM $\Aut{A}$ on input
$w \in \Sigma^*$ is a tree labeled with configurations of $\Aut{A}$,
such that the root of $T_{\Aut{A}}$ is labeled by $q_0w$, and for each
node $v$ of $T_{\Aut{A}}$ labeled by $w_1qw_2$, 
\begin{itemize}
\item if $q \in Q_\exists$, then $u$ has exactly one child, and this
  child is labeled with a successor configuration of $w_1qw_2$,
\item if $q \in Q_\forall$, then $u$ has a child $v$ for every
  successor configuration $w_1'q'w_2'$, and $v$ is labeled by
  $w_1'q'w_2'$, 
\item if $q \in \set{q_a,q_r}$, then $u$ is a leaf of $T_{\Aut{A}}$.
\end{itemize}
A computation tree is
\emph{accepting} if all its branches are finite and all its leaves are
labeled by configurations with state $q_a$. 
The \emph{language} $L(\Aut{A})$ of $\Aut{A}$ is defined as the set of all words $w \in \Sigma^*$,
for which there exists an accepting computation tree of \Aut{A} on
$w$. 
W.l.o.g., we will assume that the ATM is \emph{normalized}, i.e., every
non-halting configuration has precisely two successor configurations,
each universal step only affects the state of the machine, and  the
machine always alternates between  universal and existential states. 

The proof of Theorem~\ref{Thm:Desc_Empti_ranked_ordered}
proceeds by a reduction from the word problem for exponential-space
bounded ATMs $\Aut{A}$.
The reduction itself will be done from an ATM with empty input word.
To this end, we construct, in the canonical way, for the given exponential-space bounded ATM $\Aut{A}$ and the
given word $w\in\Sigma^*$ an ATM $\Aut{A}_w$ that works in space exponential in
the size of $w$ and accepts the empty word if, and only if,
$\Aut{A}$ accepts $w$. 
Since $\Aut{A}$ is exponential-space bounded, the non-blank portion of
the ATM's tape during a computation of $\Aut{A}_w$ will never be longer that
$2^n$, where $n$ is polynomial in the size $|w|$ of the original input.

The crucial point of the reduction is to find an encoding of
computation trees of $\Aut{A}_w$ on empty input, which can be verified by a
$\mDatalog(\tauuSigmaStrich^{\Desc})$-query that can be constructed in time
\emph{polynomial} in the size of $\Aut{A}_w$.
For this, it is necessary to find a smart encoding of the
tape inscription of length $2^n$. This encoding shall allow to compare the content of every tape
cell with the same tape cell of the successor configuration. To
achieve this, we adapt
the encoding of Björklund et al. \cite{BjorklundMS08}; in
particular, we use their very elegant
``navigation gadgets''.

We choose a fixed ranked finite alphabet $\Sigma'$ which, among other
symbols, contains
a 0-ary symbol $\bot$, 
unary symbols $r,p,m,0,1$, 
binary symbols $\CTEleft,\CTEright$, and 
3-ary symbols $\CTA$ and $s$.
Consider a computation tree $T_{\Aut{A}_w}$ of a normalized ATM
 $\Aut{A}_w=(Q,\Sigma,\Gamma,\delta,q_0)$,
see Figure~\ref{Fig:computationTree}.

\setlength{\abovecaptionskip}{0cm} 
 \setlength{\textfloatsep}{10pt plus0pt minus3pt}
\begin{figure}[thb]
\begin{center}
  \begin{tikzpicture}[->,>=stealth',  ]
  	\tikzset{
  treenode/.style = {align=center, inner sep=0pt, text centered,
    font=\sffamily},  	
  arn_red/.style = {treenode, circle, white, font=\sffamily\bfseries, draw=gray,
    fill=white, text width=1em},
  arn_blue/.style = {treenode, circle, black, draw=black,  
    text width=1em, thick},
	}  
  
 \node at (-5.3,0) {$(a)$};  
  \node at (0,0) {$(b)$}; 

    \begin{scope}[xshift=-3.5cm,scale=0.5]
    
		 \node[arn_blue] at (0,-1) (v1) {$v_1$};	
    	 \node at (2,-1) (1abel) {$w'_1q_1w''_1$};	
    
         \node[arn_blue] at (1,-2.5) (v2) {$v_2$};	
         \node at (3,-2.5) (1abel) {$w'_2q_2w''_2$};		
         
         \node[arn_blue] at (1,-4) (v3) {$v_3$};	
                \node at (3,-4) (1abel) {$w'_3q_3w''_3$};		

         \node[arn_blue] at (2,-5.5) (v4) {$v_5$};	
          \node at (4,-5.5) (1abel) {$w'_5q_5w''_5$};
         \node[arn_blue] at (0,-5.5) (v5) {$v_4$};	
          \node at (-2,-5.5) (1abel) {$w'_4q_4w''_4$};

         \node[arn_red] at (2,-7) (v8) {};	         
         \node[arn_red] at (0,-7) (v7) {};

         \node[arn_red] at (-1,-2.5) (v6) {};

        \path[draw, dotted,-] 	(0,0) edge (v1); 
	\path
    	(v1) edge (v2) 
    	    	(v2) edge (v3)
    	    	(v3) edge (v4)
    	   		(v3) edge (v5)
    	    	(v1) edge[ dotted,gray] (v6) 
	    	   	(v5) edge[ dotted,gray] (v7)
	    	   	(v4) edge[ dotted,gray] (v8) ;

    	\end{scope}      
    	
    	\begin{scope}[xshift=2cm,scale=0.5]
  	\tikzset{
  treenode/.style = {align=center, inner sep=0pt, text centered,
    font=\sffamily},  	
  arn_red/.style = {treenode, circle, white, font=\sffamily\bfseries, draw=gray,
    fill=white, text width=0.6em},
  arn_blue/.style = {treenode, circle, black, draw=black,  
    text width=0.6em, thick},
	}    
		 \node[arn_blue] at (0,-1) (v1) {};	
    	 \node[anchor=west] at (0.2,-1) (1abel) {$\CTA$};	
    
         \node[arn_blue] at (2,-2.5) (v2) {};	
         \node[anchor=west] at (2.2,-2.5) (1abel) {$\CTEright$};		
         \node[arn_blue] at (-2,-2.5) (v2_r) {};	
         \node[anchor=east] at (-2.2,-2.5) (1abel) {$r$};		
         
         \node[arn_blue] at (0,-4) (v3_r) {};	
                  \node[anchor=east] at (-0.2,-4) (1abel) {$r$};		
            \node[arn_blue] at (4,-4) (v3) {};	
                \node[anchor=west] at (4.2,-4) (1abel) {$\CTA$};		

         \node[arn_blue] at (2,-5.5) (v4_r) {};
        \node[anchor=east] at (1.8,-5.5) (1abel) {$r$};		
         \node[arn_blue] at (4,-6.5) (v4) {};	
          \node[anchor=west] at (4.2,-6.5) (1abel) {$\CTEleft$};
         \node[arn_blue] at (6,-5.5) (v5) {};	
          \node[anchor=west] at (6.2,-5.5) (1abel) {$\CTEright$};

         \node at (3.3,-7.5) (v8) {};	         
         \node at (4.8,-7.5) (v8_r) {};	         
         \node at (5,-7.5) (v7) {};	         
         \node at (7,-7.5) (v7_r) {};

         \node[arn_red] at (0,-2.5) (v6) {};

        \path[draw, dotted,-] 	(0,0) edge (v1); 
        \path
        	(v2_r) edge[-] (-1,-4.5) 
        	(-1,-4.5) edge[-] (-3,-4.5) 
        	(-3,-4.5) edge[-] (v2_r);

        \path
        	(v3_r) edge[-] (-1,-6) 
        	(-1,-6) edge[-] (1,-6) 
        	(1,-6) edge[-] (v3_r);
        \path
        	(v4_r) edge[-] (1,-7.5) 
        	(1,-7.5) edge[-] (3,-7.5) 
        	(3,-7.5) edge[-] (v4_r);

	\path
    	(v1) edge (v2) 
    	    	(v1) edge (v2_r) 
    	    	(v2) edge (v3)
    	    	(v2) edge (v3_r)    	    	
    	    	(v3) edge (v4)
	   	    	(v3) edge (v4_r)
    	   		(v3) edge (v5)
    	    	(v1) edge[ dotted,gray] (v6) 
	    	   	(v5) edge[ dotted,gray,-] (v7)
	    	   	(v5) edge[ dotted,gray,-] (v7_r)
	    	   	(v4) edge[ dotted,gray,-] (v8_r)
	    	   	(v4) edge[ dotted,gray,-] (v8) ;

 \end{scope}

\end{tikzpicture} 
\end{center}
\caption[Computation Tree]{(a)~A part of a computation tree
  $T_{\Aut{A}_w}$ where the node $v_1$ labeled by $w'_1q_1w''_1$
  is universal, and its children are existential. The node
  $v_2$ labeled by $w'_2q_2w''_2$ is the right child of $v_1$. The node
  $v_2$ has one child, the univeral node $v_3$.
  (b)~The replacement of $v_1$ is a tree with a root node labeled
  by $\CTA$ and with three children, the first is labeled by $r$ and is
  the root of the subtree encoding the configuration in $v_1$, the
  second is the replacement for its left child, and the third is the
  replacement for its right child. The obtained tree $T \deff
  \enc(T_{\Aut{A}_w})$ is an unordered ranked
  $\Sigma'$-labeled tree.}\label{Fig:computationTree} 
\end{figure}
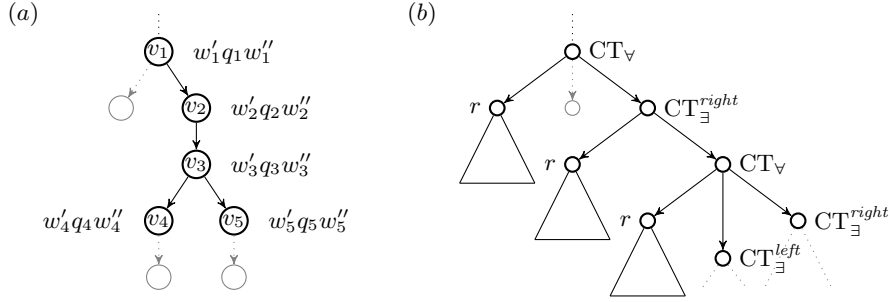

\begin{figure}[thb]
\begin{center}
  \begin{tikzpicture}[->,>=stealth']
  	\tikzset{
  treenode/.style = {align=center, inner sep=0pt, text centered,
    font=\sffamily},  	
  arn_red/.style = {treenode, circle, white, font=\sffamily\bfseries, draw=gray,
    fill=white, text width=1em},
  arn_blue/.style = {treenode, circle, black, draw=black,  
    text width=1em, thick},
	}  
  
 \node at (-5.3,0) {$(a)$};  
  \node at (0.1,0) {$(b)$}; 
    \begin{scope}[xshift=-3cm,scale=0.5]
  	\tikzset{
  treenode/.style = {align=center, inner sep=0pt, text centered,
    font=\sffamily},  	
  arn_red/.style = {treenode, circle, white, font=\sffamily\bfseries, draw=gray,
    fill=white, text width=0.6em},
  arn_blue/.style = {treenode, circle, black, draw=black,  
    text width=0.6em, thick},
	}    
		 \node[arn_blue] at (0,-1) (v1) {};	
    	 \node[anchor=west] at (0.2,-1) (1abel) {$s$};	
    
         \node[arn_blue] at (2,-2.5) (v2) {};	
         \node[anchor=west] at (2.2,-2.5) (1abel) {$s$};		
         \node[arn_blue] at (-2,-2.5) (v2_r) {};	
         \node[anchor=east] at (-2.2,-2.5) (1abel) {$s$};

         \node[arn_blue] at (0,-2.5) (v6) {};	         
         \node[anchor=west] at (0.2,-2.5) (1abel) {$p$};		
         \node[arn_blue] at (0,-4) (v6_0) {};	
         \node[anchor=west] at (0.2,-4) (1abel) {$0$};
         \node[arn_blue] at (0,-5.5) (v6_1) {};	
         \node[anchor=west] at (0.2,-5.5) (1abel) {$1$};
         \node[arn_blue] at (0,-7) (v6_leaf) {};
         \node[anchor=west] at (0.2,-7) (1abel) {$\bot$};	
         
        \path[draw, dotted,-] 	(0,0) edge (v1);

	\path
    	(v1) edge (v2) 
    	    	(v1) edge (v2_r) 
    	    	(v1) edge (v6) 
    	    	(v6) edge (v6_0)
    	    	(v6_0) edge (v6_1)
    	    	    	    	(v6_1) edge (v6_leaf)  ;   	    	
        \path
        	(v2) edge[dotted,-] (1.5,-4.5) 
        	(1.5,-4.5) edge[dotted,-] (2.5,-4.5) 
        	(2.5,-4.5) edge[dotted,-] (v2);

          \path
        	(v2_r) edge[dotted,-] (-1.5,-4.5) 
        	(-1.5,-4.5) edge[dotted,-] (-2.5,-4.5) 
        	(-2.5,-4.5) edge[dotted,-] (v2_r);

    	\end{scope}      
    	
    	\begin{scope}[xshift=2.5cm,scale=0.45]
  	\tikzset{
  treenode/.style = {align=center, inner sep=0pt, text centered,
    font=\sffamily},  	
  arn_red/.style = {treenode, circle, white, font=\sffamily\bfseries, draw=gray,
    fill=white, text width=0.6em},
  arn_blue/.style = {treenode, circle, black, draw=black,  
    text width=0.6em, thick},
	}    
		 \node[arn_blue] at (0,-1) (v1) {};	
    	 \node[anchor=west] at (0.2,-1) (1abel) {$s_{\Leaf}$};	
    
         \node[arn_blue] at (-2,-2.5) (v2_r) {};	
         \node[anchor=east] at (-2.2,-2.5) (1abel) {$p$};

         \node[arn_blue] at (2,-2.5) (v6) {};	         
         \node[anchor=west] at (2.2,-2.5) (1abel) {$m$};		
         \node[arn_blue] at (-2,-4) (v6_0) {};	
         \node[anchor=east] at (-2.2,-4) (1abel) {$1$};
         \node[arn_blue] at (-2,-5.5) (v6_1) {};	
         \node[anchor=east] at (-2.2,-5.5) (1abel) {$0$};
         \node[arn_blue] at (-2,-7) (v6_leaf) {};
         \node[anchor=east] at (-2.2,-7) (1abel) {$\bot$};

         \node[arn_blue] at (2,-4) (m) {};	         
         \node[anchor=west] at (2.2,-4) (1abel) {$0$};		
         \node[arn_blue] at (2,-5.5) (m1) {};	
         \node[anchor=west] at (2.2,-5.5) (1abel) {$1$};
         \node[arn_blue] at (2,-7) (m2) {};	
         \node[anchor=west] at (2.2,-7) (1abel) {$0$};
         \node[arn_blue] at (2,-8.5) (m_leaf) {};
         \node[anchor=west] at (2.2,-8.5) (1abel) {$\bot$};	
                 \path[draw, dotted,-] 	(0,0) edge (v1);

	\path
    	    	(v1) edge (v2_r) 
    	    	(v1) edge (v6) 
    	    	(v2_r) edge (v6_0)
    	    	(v6_0) edge (v6_1)
    	    	    	    	(v6_1) edge (v6_leaf) 
				(v6) edge (m)    	    	    	    	
				(m) edge[dotted]  (m1) 
    	    	 (m1) edge[dotted] (m2)
    	    	 (m2) edge (m_leaf);
\draw[decorate, decoration={brace}, yshift=2ex,-]  (3.5,-4) -- node[right=0.4ex] {$k$}  (3.5,-7.6);
\draw[decorate, decoration={brace}, yshift=2ex,-]  (1,-6.1) -- node[left=0.4ex] {$i$} (1,-4) ;
 \end{scope}    	
\end{tikzpicture} 
\end{center}
\caption[skeleton nodes]{(a) A skeleton node and its navigation
  gadget, indicating that the node is its parent's left child.
  (b) A skeleton node encoding a leaf of the
  configuration tree. This leaf is its parent's right child.
  It has a tape cell gadget $m$ followed by $k$ digits, the
  $i$-th of which is labeled with $1$ iff the tape cell's
  inscription is represented by the number $i$.}\label{Fig:skeletonNode}
\end{figure}
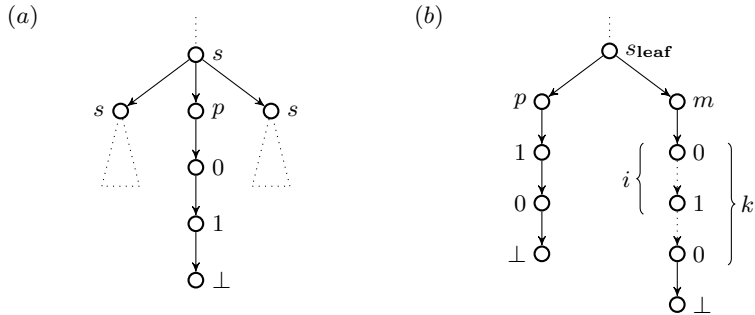
        
 We fix an arbitrary order on the
 children of nodes in $T_{\Aut{A}_w}$, such that every universal node
 has a left child and a right
 child.
The \emph{encoding} $T\deff \enc(T_{\Aut{A}_w})$ is the ranked
$\Sigma'$-labeled unordered tree obtained from $T_{\Aut{A}_w}$ by
 replacing every node $v$ labeled $w_1qw_2$ with a
 $\Sigma'$-labeled ranked tree
 $\enc(t_v)$, as follows: 
\begin{itemize}
\item if $v$ is universal, then the root of $\enc(t_v)$ is labeled with $\CTA$,
\item if $v$ is existential, and $v$ is the root of  $T_{\Aut{A}_w}$  or $v$ is the left child of a universal node, then the root of $\enc(t_v)$ is labeled with $\CTEleft$,
\item if $v$ is existential, and $v$ is the right child of
  a universal node, then the root of $\enc(t_v)$ is labeled with
  $\CTEright$, 
\item exactly one child of the root of $\enc(t_v)$ is labeled by $r$ (this will be
  the root of the subtree that encodes the configuration at $v$), and 
\item for each child $u$ of $v$ in $T_{\Aut{A}_w}$,
  $\enc(t_v)$ has a subtree $\enc(t_{u})$, which is the encoded
  subtree of $T_{\Aut{A}_w}$ obtained by the replacement of
  $u$.
\end{itemize}
\noindent 
The subtree $\gamma_r$ rooted at the $r$-labeled child of the root of $\enc(t_v)$,
encodes the configuration $c:= w_1qw_2$ represented by node $v$ in
$T_{\Aut{A}_w}$.
Since $\Aut{A}$ is exponential-space bounded, the tape inscription
of $c$ has length $\leq 2^n$. For representing $c$, we
use a full binary ordered tree
of height $n$. The path from the root to a leaf specifies the address
of the tape cell represented by the leaf, and the leaf carries
information on the tape cell's inscription and, in case that the
tape cell is the current head position, also information on the
current state; all this information is encoded by a suitable
\emph{tape cell gadget} that is attached to the ``leaf''. 
The number $k$ of possible tape cell inscriptions (enriched
with information on the current state) is \emph{polynomial} in $\size{\Aut{A}_w}$.
The nodes of the ``full binary tree'' are
called \emph{skeleton nodes} and are labeled $s$.
To ensure that the desired query $Q$ can be constructed in \emph{polynomial}
time, we attach to each skeleton node a \emph{navigation gadget} \cite{BjorklundMS08}, which is a
path of length~4. To indicate that a node is a
left (resp., right) child, this gadget is labeled $p - 0 - 1 - \bot$ (resp.,
$p - 1 - 0 - \bot$). See Figure~\ref{Fig:skeletonNode} for an
illustration of the \emph{navigation gadget} and the \emph{tape cell gadget}. 
    
Given an ATM $\Aut{A}$ and a word $w\in\Sigma^*$, we construct 
in polynomial time an $\mDatalog(\tau_{u,\Sigma'}^{\Desc})$-query $Q=(\PP,\Ans)$ such that 
\,$Q_{\textit{Bool}} \neq \emptyquery$\, iff \,there is an
accepting computation tree for $\Aut{A}_w$ on $\emptyword$, i.e., \,$w \in L(\Aut{A})$. 
The query $Q$ consists of two parts, one to
verify that the structure of the input tree represents an encoded
computation tree, and the other to verify consistency with the ATM's 
transition relation.
Details can be found in the appendix.
The particular choice of the navigation gadgets ensures that
$Q$ can be constructed in time \emph{polynomial} in the size of
$\Aut{A}$ and $w$. The only point where we make essential use of the $\Desc$-predicate is during the comparison of the cells by  using the navigation gadgets.

\section{Final Remarks}\label{section:relatedResults}

Along with the upper bound provided by Theorem~\ref{thm:AMW14}, and 
since $\tauu^{\Desc}\subseteq \tauo^{\Child,\Desc}$,
Theorem~\ref{Thm:Desc_QCP_unranked_unordered} implies the following
corollary, which summarizes our main results.

\begin{Corollary}\label{cor:main}
 The QCP is $\TwoEXPTIME$-complete for 
 Boolean $\mDatalog(\tauu^{\Desc})$ on finite labeled unordered trees, and for
 Boolean $\mDatalog(\tauo^{\Child,\Desc})$ on finite labeled ordered trees.
\end{Corollary}
By applying standard reductions, the $\TwoEXPTIME$-completeness results
of Corollary~\ref{cor:main} carry over from the QCP to the
\emph{equivalence problem}.
When restricting attention to \emph{ranked} trees over a ranked finite
alphabet, the $\TwoEXPTIME$-completeness results also carry over to the \emph{emptiness problem}.
For \emph{unranked} labeled trees, the emptiness problem for
$\mDatalog(\tauo^{\Child,\Desc})$ is in $\TwoEXPTIME$, but we
currently do not have a matching $\TwoEXPTIME$-hardness result.

\setlength{\belowcaptionskip}{0pt} 
\setlength{\intextsep}{2pt} 
An overview of the currently known results is given in Table~\ref{tbl:results}; for further information and detailed proofs we refer to \cite{frochaux}.
\begin{table}
\caption{Complexity of monadic datalog on finite labeled trees; $N \subseteq \set{\Root,\Leaf}$ and $M \subseteq \set{\Root,\Leaf,\Ls,\Child}$; ``c'' (``h'') means ``complete'' (``hard'').} \label{tbl:results}
\centering
\begin{tabular}{|l|c|c|c|c|c|c|}
\hline	
 & $\tauu^N$	& $\tauo^M$  & $\tauu^{N\cup\set{\Desc}}$  &  $\tauo^{M\cup\set{\Child,\Desc}}$ &  $\tauGK^{\Child,\Desc}$ & \\
\hline \hline
& \multicolumn{2}{c|}{} & \multicolumn{2}{c|}{\EXPTIME-h {\scriptsize \&} in \TwoEXPTIME}  &&  \textit{unranked} \\ \cline{4-5}
\raisebox{1.5ex}[-1.5ex]{Emptiness}  & \multicolumn{2}{c|}{\raisebox{1.5ex}[-1.5ex]{\EXPTIME-c}} &\multicolumn{3}{c|}{\TwoEXPTIME-c} & \textit{ranked} \\
\hline
 &  \multicolumn{2}{c|}{} & \multicolumn{3}{c|}{} & \textit{unranked}\\
\raisebox{1.5ex}[-1.5ex]{Equivalence} & \multicolumn{2}{c|}{\raisebox{1.5ex}[-1.5ex]{\EXPTIME-c}} & \multicolumn{3}{c|}{\raisebox{1.5ex}[-1.5ex]{\TwoEXPTIME-c}} & \textit{ranked} \\
\hline
 & \multicolumn{2}{c|}{} & \multicolumn{3}{c|}{} & \textit{unranked}\\
\raisebox{1.5ex}[-1.5ex]{Containment} & \multicolumn{2}{c|}{\raisebox{1.5ex}[-1.5ex]{\EXPTIME-c}}  & \multicolumn{3}{c|}{\raisebox{1.5ex}[-1.5ex]{\TwoEXPTIME-c}} &\textit{ranked} \\
\hline
\end{tabular}
\end{table}

\clearpage

\bibliographystyle{splncs03}
\bibliography{MonadicDatalogContainmentOnTreesUsingDesc}

\appendix
\newcounter{restoreAppTheorem}
\section*{APPENDIX}

  \section{Hardness on Ranked  Trees}

\noindent{\hypersetup{pdfborderstyle={/S/U/W 0}}\textbf{Theorem \ref{Thm:Desc_Empti_ranked_ordered}  \textit{(restated)}} \emph{ There is a ranked finite alphabet $\Sigma'$, such that 
 the emptiness problem for Boolean $\mDatalog(\tauuSigmaStrich^{\Desc})$ on
 finite unordered ranked $\Sigma'$-labeled trees is \TwoEXPTIME-hard.} 
\\ }

The proof idea of Theorem~\ref{Thm:Desc_Empti_ranked_ordered} is based on the proof of the Theorem~\ref{Thm:Validit} presented in the full version of the MFCS publication by Björklund, Martens, and Schwentick \cite{BjorklundMS08}. The used alternating Turing machine was introduced at FOCS'76 by Chandra and Stockmeyer, as well as by Kozen, 
and presented in a  joint journal publication in 1981\cite{DBLP:journals/jacm/ChandraKS81}.

An \emph{alternating Turing machine} (ATM, for short) $ \Aut{A}=(Q,\Sigma,\Gamma,\delta,q_0)$ consists of \begin{itemize}
\item a finite set of states $Q$ partitioned into \emph{universal states} $Q_{\forall}$, \emph{existential states} $Q_{\exists}$, an \emph{accepting} state $q_a$, and a \emph{rejecting state} $q_r$,
\item the finite input alphabet $\Sigma$,
\item the finite tape alphabet $\Gamma \supset \Sigma$, that contains the special \emph{blank} symbol $\Blank$,
\item the initial (or, starting) state $q_0$ and
\item the transition relation $\delta \subseteq ((Q \times \Gamma) \times (Q \times \Gamma \times \set{L,R,S}))$.
 \end{itemize} As usual the letters $L$, $R$, and $S$ denote the directions \textit{left}, \textit{right}, and \textit{stay} in which the head on the tape is moved.
 
 A \emph{configuration} $c$ of $\Aut{A}$ is given by specifying its state, the content of its tape together with the position of the tape head. Thus, we interpret a string of the form $w_1qw_2$ with $w_1, w_2 \in \Gamma^*$, $q \in Q$ as the configuration in which the tape contains the word $w_1w_2$, followed by blanks, the head's tape  position is the first letter of $w_2$, and $q$ is the current state of the machine. A transition rule $((q,a),(q',b,D)) \in \delta$ denotes a step of $\Aut{A}$ by reading in state $q$ the letter $a \in \Gamma$, overwriting $a$ on the current head position by $b \in \Gamma$, moving the head depending on $D \in \set{L,R,S}$ one position to the left, to right, or stay, and finally, switching to state $q'$. A configuration $c'$ obtained by applying a  rule of $\delta$ to a given configuration $c$ is called \emph{successor configuration} of $c$. The configuration $w_1bq'w_2$, for example, is a successor configuration of $w_1qaw_2$ obtained by appling the transition rule  $((q,a),(q',b,R))$. A configuration $w_1qw_2$ is a \emph{halting configuration} if $q$ is either the accepting state $q_a$ or the rejecting state $q_r$. Without loss of generality, we can assume that there is no successor configuration of any halting configuration, and furthermore, before halting, the automaton moves its head to the left on the first non-blank symbol on the tape, so  each halting configuration is of the form $qw$. 

A \emph{computation tree} $T_{\Aut{A}}$ of the ATM $\Aut{A}$ on input $w \in \Sigma^*$ is a tree labeled with configurations of $\Aut{A}$, such that the root of $T_{\Aut{A}}$ is labeled by $q_0w$, and for each node $u$ of $T_{\Aut{A}}$ labeled by $w_1qw_2$,
\begin{itemize}
\item if $q \in Q_\exists$, then $u$ has exactly one child, and this child is labeled with a successor configuration of $w_1qw_2$,
\item if $q \in Q_\forall$, then $u$ has a child $v$ for every successor configuration $w_1'q'w_2'$ of $w_1qw_2$, and $v$ is labeled by $w_1'q'w_2'$, 
\item if $q \in \set{q_a,q_r}$, then $u$ is a leaf of $T_{\Aut{A}}$.
\end{itemize}
Observe, that $T_\Aut{A}$ can be infinite, since $\Aut{A}$ may have non-halting computation branches. A computation tree is \emph{accepting} if all its branches are finite and all its leaves are labeled by configurations in state $q_a$. As usually, the language $L(\Aut{A})$ of the ATM \Aut{A} is the set of words $w \in \Sigma^*$ for which there exists an accepting computation tree of \Aut{A} on $w$.

We say that an ATM is \emph{normalized} if every non-halting configuration has precisely two successor configurations, each universal step only affects the state of the machine, and additionally, the machine always proceeds from an universal state to an existential state, and vice versa. It is easy to verify that for every alternating Turing machine $\Aut{A}$ there exists a normalized alternating Turing machine $\Aut{A}_n$ with $L(\Aut{A}) =  L(\Aut{A}_n)$, and $\Aut{A}_n$ can be constructed from  $\Aut{A}$ within polynomial time.

Now, we are ready to prove Theorem~\ref{Thm:Desc_Empti_ranked_ordered}.

\begin{proofof}{Theorem~\ref{Thm:Desc_Empti_ranked_ordered}}
Our proof proceeds by a reduction from the word problem for exponential space bounded ATM $\Aut{A}$. In this problem, the input consists on an exponential space bounded ATM $\Aut{A}$, and an input word $w$ for $\Aut{A}$, and the task is to decide if $w \in L(\Aut{A})$. In \cite{DBLP:journals/jacm/ChandraKS81} this problem was shown to be $\TwoEXPTIME$-complete.  

Our reduction will be done from an ATM with empty input. Therefore, we construct for the given ATM $\Aut{A}$ and the given word $w$ an ATM $\Aut{A}_w$ that works in space exponential in the size $w$ and accepts the empty input word if and only if $\Aut{A}$ accepts $w$. To do this, we let $\Aut{A}_w$ start by writing $w$ on the empty tape, afterwards $\Aut{A}_w$ returns to the leftmost tape position and finally, it starts to simulate the original machine $\Aut{A}$. W.l.o.g., we can assume that $\Aut{A}_w$ is normalized and since the computation is exponentially space bounded, the non-blank portion of the tape during the computation of $\Aut{A}_w$ is never longer that $2^n$, where $n$ is polynomial in the size $|w|$ of the original input word.

We will choose a suitable ranked alphabet $\Sigma'$, independent from $\Aut{A}_w$. Within polynomial time, we construct an $\mDatalog(\tau_{u,\Sigma'}^{\Desc})$-query $Q=(\PP,\Ans)$ such that \begin{align*} Q \neq \emptyset \qquad &\iff \qquad \text{there is an accepting computation tree for $\Aut{A}_w$}\\
 &\iff \qquad w \in L(\Aut{A}). \end{align*}
 
Since $\TwoEXPTIME$ is closed under complement, it implies that the emptiness problem for Boolean $\mDatalog(\tauu^{\Desc})$ on ranked unordered labeled trees is hard for \TwoEXPTIME. 
 
In the next paragraphs, we present the encoding of the computation tree that is basically taken from \cite{BjorklundMS08} and includes the encoding of the configuration tree, both are adapted to our problem.  So, let $T_{\Aut{A}_w}$ be a computation tree of $\Aut{A}_w=(Q,\Sigma,\Gamma,\delta,q_0)$ (cf. Figure~\ref{Fig:computationTree}), we fix some arbitrary order of the children of each universal node  such that every universal node has a left and a right child.\footnote{A  node $u$ is an universal node if it is labeled by a configuration $w_1qw_2$ where $q$ is an universal state. If $q$ is an existential state then $u$ is existential.} Now the encoding $T \deff \enc(T_{\Aut{A}_w})$ can be obtained from $T_{\Aut{A}_w}$ by replacing every node $v$ labeled by $w_1qw_2$ with a tree $\enc(t_v)$, as follows \begin{itemize}
\item if $v$ is universal, then the root of $\enc(t_v)$ is labeled with $\CTA$,
\item if $v$ is existential, and $v$ is the root of  $T_{\Aut{A}_w}$  or $v$ is the left child of a universal node, then the root of $\enc(t_v)$ is labeled with $\CTEleft$,
\item if $v$ is existential, and $v$ is the right child of
  a universal node, then the root of $\enc(t_v)$ is labeled with
  $\CTEright$, 
\item exactly one child of the root of $\enc(t_v)$ is labeled by $r$ (this will be
  the root of the subtree that encodes the configuration at $v$), and 
\item for each child $u_i$ of $v$ in $T_{\Aut{A}_w}$, $\enc(t_v)$ has a subtree $\enc(t_{u_i})$, which is the encoded subtree of $T_{\Aut{A}_w}$ obtained by the replacement of $u_i$.\footnote{In fact, for a non halting configuration there is exactly one child if $v$ is existential or otherwise, if $v$ is universal, there are exactly two children since $\Aut{A}_w$ is normalized.}
\end{itemize}

\noindent The  set of subtrees denoted by their root label $r$ encode the configurations that is originally labeled the computation tree. We have to navigate through $2^n$ tape cells and we must be able to compare the $i$-th cell of one configuration with the $i$-th cell of the predecessor configuration. Thus, the configuration tree is basically a binary tree of height $n$ that has $2^n$ leaves to carry the information for the tape cells, together with the information of the current state of the machine and the position of the head. This sequence of $2^n$  \emph{configuration cells} will carry the whole information about the configuration of the machine in this working step. To this end, the set of configuration cells is partitioned into three types.
\begin{itemize}
\item The set \BCells\ of \emph{basic cells} is equal to $\Gamma$. A basic cell represents a tape cell that is not currently visited by the head and also is not visited in the predecessor configuration.
\item The set \CCells\ of \emph{current tape head cells} is equal to $\Gamma \times \delta$. The letter from $\Gamma$ represents the tape content in the actual position that is currently visited by the head, while the transition from $\delta$ is the transition which leads to the actual configuration.
\item The set \PCells\ of \emph{previous tape head cells} is equal to $\Gamma \times (Q \times \Gamma)$ and represent tape cells that were visited by the head in the predecessor configuration, but not in the current one. The first letter from $\Gamma$ represents the actual content on the tape in this cell and the pair $(Q \times \Gamma)$ the previous state and tape content in the predecessor configuration.
\end{itemize}
Observe, the number $k$ of all possible configuration cells for $\Aut{A}_w$ is polynomial in the size of the automaton and so we can refer to each possible  configuration cell a natural number $i$ in $\set{1,\ldots,k}$.

Now, it is necessary to fix a set of constraints, that allows to decide whenever a sequence $C_1$ of $2^n$ configuration cells is a valid successor configuration of another sequence $C_0$. We start with constraints to ensure a degree of consistency inside a given sequence. The set $H(\Aut{A}_w)$ of horizontal constraints consists of the following rules:
\begin{enumerate}[(H1)]
\item The only cell allowed to the left of a cell $(a,((q_1,b),(q_2,c,R))) \in \CCells $ is the cell $(c,(q_1,b)) \in \PCells$.
\item The only cell allowed to the right of a cell $(a,((q_1,b),(q_2,c,L))) \in \CCells $ is the cell $(c,(q_1,b)) \in \PCells$.
\item The only cell allowed to the right of the basic cell $\Blank \in \Gamma$ is $\Blank$ itself.
\end{enumerate}
To fix the set $V(\Aut{A}_w)$ of vertical constraints between two  consecutive sequences $C_0$ and $C_1$, we imagine the predecessor is lying cell by cell on top of its successor such that the $i$-th configuration cell of $C_0$ is lying on top of the $i$-th cell of $C_1$.
\begin{enumerate}[(V1)]
\item If the $i$-th cell is a BCell $a \in \Gamma$ then the only allowed cells on the $i$-th tape position in a successor configuration are $a$ itself and any CCell $(a,((q_1,b),(q_2,c,m))$ where $m \in \set{L,R}$. The latter is the case that the automaton $\Aut{A}_w$ just moved to this cell, coming from the left or the right. The letter on this position is currently untouched, but the letter in the left (right) neighbor is overwritten if $b\neq c$ and $m=L$ ($m=R$).
\item   If the $i$-th cell is a CCell $(a,((q_1,b),(q_2,c,m)))$ then the only allowed cells on the $i$-th tape position in a successor configuration are any $(d,(q_2,a)) \in \PCells$ and any $(d,((q_2,a),(q_3,d,m'))) \in \CCells$ where $m'=S$.
\item  If the $i$-th cell is a PCell $(a,(q,b))$ then the only allowed cells on the $i$-th tape position in a successor configuration are the BCell $a$ and any CCell $(a,((q_1,b),(q_2,c,m))$ where $m \in \set{L,R}$.
\end{enumerate}
Figure~\ref{Fig:configATM} illustrates an example of valid transitions respecting this constraints. It is easy to verify that if $C_0$ is a valid encoding of a configuration, $C_1$ is a valid encoding of a successor configuration if and only if all horizontal and vertical conditions are satisfied.
\begin{figure}[h]
\begin{center}
  \begin{tikzpicture}

    \begin{scope}[scale=0.7]

		\node[anchor=  west] at (-1.8,2.5) {$C_{x}$};
         \path[draw] 	(0,0) edge (-1,0); 
                  \path[draw] 	(0,2) edge (-1,2); 
                     \path[draw] 	(13,0) edge (12,0); 
                  \path[draw] 	(13,2) edge (12,2); 
                  		\node at (13.1,1) {$\ldots$};
                  		\node at (-1.1,1) {$\ldots$};
        \draw
			(0, 0) rectangle (4, 2);
		\node[anchor= north west] at (0,2) {BCell};
		\node at (2,1) {$e$};

        \draw
			(4, 0) rectangle (8, 2);
		\node[anchor= north west] at (4,2) {CCell};
		\node at (6,1) {$b$};
		\node[anchor= south] at (6,0) {$((q,c),(q',f,L))$};

        \draw
			(8, 0) rectangle (12, 2);
		\node[anchor= north west] at (8,2) {PCell};
		\node at (10,1) {$f$};
		\node[anchor= south] at (10,0) {$(q,c)$};

   \begin{scope}[yshift=-3.5cm]

		\node[anchor=  west] at (-1.8,2.5) {$C_{x+1}$};
         \path[draw] 	(0,0) edge (-1,0); 
                  \path[draw] 	(0,2) edge (-1,2); 
                     \path[draw] 	(13,0) edge (12,0); 
                  \path[draw] 	(13,2) edge (12,2); 
                  		\node at (13.1,1) {$\ldots$};
                  		\node at (-1.1,1) {$\ldots$};
        \draw
			(0, 0) rectangle (4, 2);
		\node[anchor= north west] at (0,2) {BCell};
		\node at (2,1) {$e$};

        \draw
			(4, 0) rectangle (8, 2);
		\node[anchor= north west] at (4,2) {CCell};
		\node at (6,1) {$a$};
		\node[anchor= south] at (6,0) {$((q',b),(q'',a,S))$};
		
        \draw
			(8, 0) rectangle (12, 2);
		\node[anchor= north west] at (8,2) {BCell};
		\node at (10,1) {$f$};

    	\end{scope}

    	\end{scope}

\end{tikzpicture} 
\end{center}
\caption[Example of a sequence of configurations]{This example shows the corresponding parts of a valid configuration $C_{x}$ and its successor configuration $C_{x+1}$. The previous transition $((q,c),(q',f,L))$ leading to configuraion $C_x$ was reading a $c \in \Gamma$ on the right cell, writing an $f \in \Gamma$, switching the state from $q$ to $q'$, and finally moving the head one position to the left. The changeover from $C_x$ to $C_{x+1}$ was done by using transition $((q',b),(q'',a,S))$, saying reading in state $q'$ the letter $b$, write the letter $a$, switch to state $q''$, and stay with the head at the current position.}\label{Fig:configATM}
\end{figure}
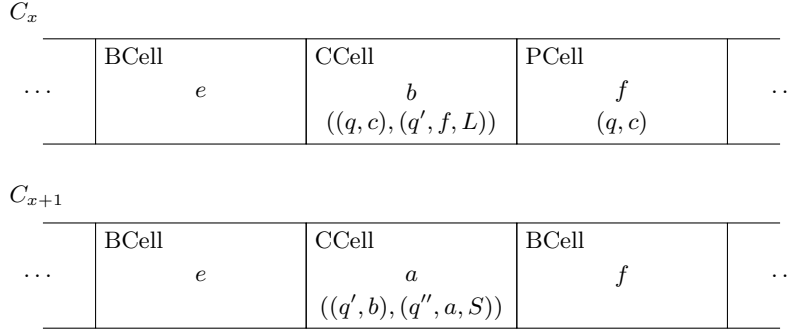

Now, we are ready to describe the structure of the $r$-rootet subtrees that encode the configuration; that is the last remaining part of the whole encoding. We already noted that these configuration trees are based on binary trees of height $n$. Every non root node carries the label $s$ and Björklund et al. called them \emph{skeleton nodes}.  Every skeleton node has an attached \emph{navigation gadget}, that is a short path of four nodes labeled by $p$, $0$, $1$, $\bot$ for denoting any children as left children and labeled by $p$, $1$, $0$, $\bot$ for right children in the sequence from the skeleton node to the leaf of the gadget (cf. Figure~\ref{Fig:skeletonNode} (a)).

Each leaf skeleton node, that is a skeleton node that has no skeleton node as child, carries besides the navigation gadget, a \emph{configuration cell gadget} that consists of a path of length $k+2$.\footnote{Recall, $k$ is the number of all possible configuration cells of $\Aut{A}_w$.} The root node of this path is labeled by $m$ (for \emph{me}) followed by $k$ nodes labeled with digits $0$ and $1$, and the path ends in a leaf labeled with $\bot$. $k-1$ nodes on this path are labeled with $0$, only the $i$-th node is labeled by $1$, telling the current cell is the cell number $i$. 

To finish the description of the encoding, for technical reasons, we start in the top of the computation tree with a node labeled with $\top$ that has exactly one child, the topmost configuration node. Now, we are ready to define the ranked alphabet $\Sigma'$ and afterwards, to construct the query. The alphabet consists of the following symbols:
\begin{quote}
\begin{description}
	\item[$\top$] of arity $ar(\top)=1$, that denotes the root node of the encoded computation tree.
   \item[$\CTA$] of arity $\ar(\CTA)=3$, that denotes a universal configuration.
   \item[$\LeafCTA$] of arity $\ar(\LeafCTA)=1$, that denotes a  halting configuration, that is a child of an existential configuration.\footnote{To be precise, a halting configuration is neither an existential nor a universal configuration, but the labels tell us whose configuration child it is.}
   \item[$\CTEleft$] of arity  $\ar(\CTEleft)=2$, that denotes  an existential configuration  where the configuration itself is the left child of a universal configuration (or the initial configuration).
      \item[$\CTEright$] of arity $\ar(\CTEright)=2$, that denotes an existential configuration  where the configuration itself is the right child of a universal configuration.
   \item[$\LeafCTEleft$] of arity  $\ar(\LeafCTEleft)=1$, that denotes a halting configuration  where the configuration itself is the left child of a universal  configuration (or the initial configuration).
    \item[$\LeafCTEright$] of arity $\ar(\LeafCTEright)=1$, that denotes a halting configuration  where the configuration itself is the right child of a universal  configuration.
   \item[$r$] of arity $\ar(r)=2$, that denotes the root node of an configuration tree.
    \item[$s$] of arity  $\ar(s)=3$, that denotes a skeleton node of an configuration tree.
      \item[$s_{\Leaf}$] of arity  $\ar(s_{\Leaf})=2$, that denotes a skeleton leaf node that is a leaf of the configuration tree.
   \item[$p$] of arity $\ar(p)=1$, that denotes the root of an navigation gadget.
   \item[$m$] of arity $\ar(m)=1$, that denotes the root of an cell gadget 'me'.
    \item[$0$ and $1$] of arity $\ar(0)=\ar(1)=1$, for the values of the gadgets.
    \item[$\bot$] the only symbol of $\Sigma_T$ of arity $\ar(\bot)=0 $. So every leaf of the encoding tree is labeled by $\bot$.
\end{description}
\end{quote}
The construction of the demanded query $Q=(\PP,Ans) $ starts with a program $\PP_1$ that ensures the newly introduced $\idb$ predicate $\structure$ for the root if the input tree $T$ is structured as an encoded computation tree. In particular, the input tree must fulfill the following conditions.
\begin{enumerate}[(1)]
\item \label{cond:kTOP}The root of the tree is labeled with $\top$ and has exactly one child that represents the initial configuration.
\item \label{cond:CT-r}Each configuration node has exactly one child labeled with $r$.
\item \label{cond:vGadget} Every configuration cell gadget correctly encodes a configuration cell.
\item Each encoded configuration tree is complete and has height $n$.
\item \label{cond:nGadget} Every skeleton node has exactly one correctly assigned navigation gadget.
\item \label{cond:H} All horizontal constraints from $H(\Aut{A}_w)$ are satisfied.
\item The universal and existential configurations must alternate on the subtree of $\CT$ labeled nodes.
\item \label{Cond:diffCCells}For each non halting universal configuration, the two child configuration nodes represent two encoded configuration trees with two different CCells.
\item \label{Cond:startCT}The highest encoded configuration tree has the start configuration cell \[(\Blank, ((q_0,\Blank),(q_0,\Blank,S)))\] as its leftmost configuration cell.
Recall, $q_0$ is the initial state of $\Aut{A}_w$ and the computation starts on an empty tape.
\item \label{cond:CTleafAcc}Every configuration node that has no successor configuration encodes a final configuration, that implies the leftmost configuration cell is of the form \[ (a,((q,b),(q_a,c,m))).\] Recall $q_a$ is the accepting state of the machine, the machine, upon accepting, moves its head to the leftmost tape cell, and finally, an input tree is accepted if every path in the computation tree leads to an accepting halting configuration.
\end{enumerate}
The program $\PP_1$ will start in the leaves of the encoded tree and verifies the structure step by step in the direction to the root node. For the beginning, the program $\PP_1$ is the empty set of rules and the first rule we add is  to call leaves by what they are. Thus, we add 
\[ \Leaf(x) \leftarrow \Label_{\bot}(x).\]
In any case, a leaf belongs to a gadget, that is a cell configuration or a navigation gadget, and therefore we count the length of the digit path up to the length of $k$ by the following rules.
			\begin{align*}
				0(x) & \leftarrow \Label_{0}(x)\\
				1(x) & \leftarrow \Label_{1}(x)\\
				\digit(x) & \leftarrow 0(x)\\
				\digit(x) & \leftarrow 1(x)\\
				\digit_0(x) & \leftarrow \Leaf(x) \\
				\digit_1(x) & \leftarrow \digit(x),\Child(x,y),\digit_0(y) \\
				\digit_2(x) & \leftarrow \digit(x),\Child(x,y),\digit_1(y) \\
				\ldots \quad \\
				\digit_k(x) & \leftarrow \digit(x),\Child(x,y),\digit_{k-1}(y) 
			\end{align*}		
Additionally, to ensure that a navigation gadget and the 'me' 
 cell configuration gadget have exactly one node labeled with $1$, we count the amount of $1$-labeled nodes on every digit path by the following rules.
			\begin{align*}
				\Ccount^1_{<1}(x) & \leftarrow \Leaf(x) \\
				\Ccount^1_{<1}(x) & \leftarrow  0(x),\Child(x,y),\Ccount^1_{<1}(y) \\
				\Ccount^1_{=1}(x) & \leftarrow  1(x),\Child(x,y),\Ccount^1_{<1}(y) \\
				\Ccount^1_{=1}(x) & \leftarrow  0(x),\Child(x,y),\Ccount^1_{=1}(y) \\
			\end{align*}		
We propagate this counting results to the gadget roots if they are labeled by $m$ or $p$ by adding the following rules to $\PP_1$:
			\begin{align*}
				\Ccount^1_{=1}(x) & \leftarrow  \Label_{p}(x),\Child(x,y),\Ccount^1_{=1}(y) \\
				\digit_2(x) & \leftarrow  \Label_{p}(x),\Child(x,y),\digit_2(x) \\
				p(x)& \leftarrow  \Label_{p}(x), \Ccount^1_{=1}(x),  \digit_2(x) \\
				\Ccount^1_{=1}(x) & \leftarrow  \Label_{m}(x),\Child(x,y),\Ccount^1_{=1}(y) \\
				\digit_k(x) & \leftarrow  \Label_{m}(x),\Child(x,y),\digit_k(x) \\
				m(x)& \leftarrow   \Label_{m}(x), \Ccount^1_{=1}(x),  \digit_k(x) 
			\end{align*}	
Now, the predicate $p$  becomes true for a node $v$ of the input tree $T$ if it is labeled with $p$ and it is the starting node of a navigation gadget that actually denotes a direction, as well as, $m$ becomes true for a node $v$ of the input tree $T$ if it is labeled with $m$ and it is the starting node  of  a 'me'  cell configuration gadget that actually denotes a configuration cell.

For the rest of the section, we introduce a predicate $\Child^i(x,y)$ for a natural number $i$ as short hand for the set of atoms \[ \Child(x,x_1),\Child(x_1,x_2), \ldots ,\Child(x_{i-1},y) \] where $\Child^i(x,y)$ states  the fact that $y$ is a descendant of $x$ in the $i$-th generation.

By the following rules, every $m$-marked node knows which configuration $i \in \set{1,\ldots k}$ it encodes.
\begin{align*}
	m_{k=1}(x) & \leftarrow m(x),\Child(x,x_1), 1(x_1) \\
	m_{k=2}(x) & \leftarrow m(x),\Child^2(x,x_2), 1(x_2) \\
		\vdots \quad & \\
		m_{k=i}(x) & \leftarrow m(x),\Child^i(x,x_i), 1(x_i) \\
		\vdots \quad & \\
	m_{k=k}(x) & \leftarrow m(x),\Child^k(x,x_k), 1(x_k) \\
\end{align*}
	
	Now, we mark the leaves of the skeleton nodes with the $\idb$ predicate $s_{\Leaf}$ that are leaves in the configuration tree considered without the gadgets.
\[
		s_{\Leaf}(x) \leftarrow \Label_{s_{\Leaf}}(x), \Child(x,x_m), m(x_m), \Child(x,x_p), p(x_p) 
	\] Observe, the label $s_{\Leaf}$ has arity two, so there cannot be  further children the rule could work on. Now, for the subtrees rooted by nodes marked with $s_{\Leaf}$ the condition (\ref{cond:vGadget}) is fulfilled. 
	By the next rules, we mark the nodes carrying the label $s$ or  $s_{\Leaf}$ regarding their navigation gadget as left child using $s_L$ or as right child by using $s_R$. Remember a correct navigation gadget is marked by the $\idb$ predicate $p$.
	\begin{align*}
		s_L(x) & \leftarrow \Label_{s_{\Leaf}}(x), \Child(x,x_p), p(x_p), \Child(x_p,x_n), 0(x_n) \\
		s_R(x) & \leftarrow \Label_{s_{\Leaf}}(x), \Child(x,x_p), p(x_p), \Child(x_p,x_n), 1(x_n) \\
		s_L(x) & \leftarrow \Label_s(x), \Child(x,x_p), p(x_p), \Child(x_p,x_n), 0(x_n) \\
		s_R(x) & \leftarrow \Label_s(x), \Child(x,x_p), p(x_p), \Child(x_p,x_n), 1(x_n) 
	\end{align*}
	We are going to mark the entire configuration tree with the predicate $s$, that affects the nodes marked by $s_{\Leaf}$ and every node labeled by $s$ that have a correct navigation gadget, as well as  left and  right children.
		\begin{align*}
			s(x) & \leftarrow s_{\Leaf}(x)\\
			s(x) & \leftarrow s_L(x), \Child(x,x_l),s_L(x_l),s(x_l),\Child(x,x_r),s_R(x_r),s(x_r)\\
			s(x) & \leftarrow s_R(x), \Child(x,x_l),s_L(x_l),s(x_l),\Child(x,x_r),s_R(x_r),s(x_r)
		\end{align*}
		Note, an inner node of the configuration tree is itself a left or right child, that implies there is such a navigation gadget and it gets the $s$ predicate, if it has a left and a right child, marked with $s_L$ and $s_R$. This implies, this node cannot own a second navigation gadget that claims the opposite of another navigation gadget since the arity of the symbol $s$ enforces the limit of exactly three children. Remember, we have to ensure that the configuration tree is complete and has height $n$. This will be done if both children of the $r$ labeled root of the configuration tree are marked by height $n-1$ and by $s$ since $s$ is only true for them if every $s$ child itself has two $s$ children downto the leaves of the configuration tree. So, up to $n-1$, we count the height of the configuration tree by adding the following rules to $\PP_1$.
		\begin{align*}
			s_{h=0}(x) & \leftarrow s_{\Leaf}(x)\\
			s_{h=1}(x) & \leftarrow \Child(x,x_l), s_L(x_l), s_{h=0}(x_l) , \Child(x,x_r),s_R(x_r), s_{h=0}(x_r) \\
			\vdots \quad &\\
			s_{h=n-1}(x) & \leftarrow \Child(x,x_l), s_L(x_l), s_{h=n-2}(x_l) , \Child(x,x_r),s_R(x_r), s_{h=n-2}(x_r) 
		\end{align*}
To this end, we mark a node labeled by $r$ with the predicate $r_{nav}$ if it is the root of a navigable and  complete configuration tree and add the rule
\begin{align*} r_{nav}(x) &\leftarrow \Label_r(x), \Child(x,x_l), s_L(x_l),s(x_l), s_{h=n-1}(x_l), \qquad \\& \pushright{\Child(x,x_r), s_R(x_r),s(x_r), s_{h=n-1}(x_r)}
 \end{align*}
to $\PP_1$. Observe, during the computation of $\PP_1(T)$ a node labeled by $r$ gets marked with $r_{nav}$ if it is a root of a complete configuration tree of height $n$ where every skeleton node carries a correct navigation gadget and in the skeleton leaves a cell configuration 
is  correctly encoded. So, the conditions (\ref{cond:vGadget}) -- (\ref{cond:nGadget}) are fulfilled. 

The next goal is to ensure condition (\ref{cond:H}) that stands for the horizontal constraints (H1)--(H3). This actually holds if the tuple  $(i,j)$  of two neighboring configurations cells is contained in the relation $H(\Aut{A}_w)$. Remember, a node labeled by $m$ is already marked by $m_{k=i}$ for its encoded configuration $i$. In a first step and for every $i\in \set{1,\ldots,k}$, we propagate this information to the skeleton leaves by the following rules.
\[
	(k=i)_{\Leaf}(x)  \leftarrow s_{\Leaf}(x), \Child(x,y), m_{k=i}(y)
\]
Next, we propagate for a subtree of the configuration tree its leftmost and its rightmost configuration cell. Furthermore, it is to verify if the rightmost cell of the left child fits together with the leftmost cell of the right child. Therefore, we use the new predicates  $(k=i)_{\Left}$ and $(k=i)_{\Right}$ for every $i \in \set{1,\ldots, k}$ in the following rules 
\begin{align*}
	(k=i)_{\Left}(x) &\leftarrow  s(x),(k=i)_{\Leaf}(x) \\
	(k=i)_{\Right}(x) &\leftarrow  s(x),(k=i)_{\Leaf}(x) \\	
	(k=i)_{\Left}(x) &\leftarrow  s(x),\Child(x,x_l), s_L(x_l),(k=i)_{\Left}(x_l) \\
	(k=i)_{\Right}(x) &\leftarrow  s(x),\Child(x,x_r), s_R(x_r),(k=i)_{\Right}(x_r) ,
\end{align*} as well as for every $(i,j)\in H(\Aut{A}_w)$,  the predicate $H$ (if the nodes children fit together) in the following rules 
\begin{align*}
	H(x) & \leftarrow s(x), \Child(x,x_l), s_L(x_l),s_{\Leaf}(x_l),(k=i)_{\Right}(x_l),\\ & \pushright{\Child(x,x_r), s_R(x_r),s_{\Leaf}(x_r),(k=j)_{\Left}(x_r)}\\
	H(x) & \leftarrow s(x), \Child(x,x_l), s_L(x_l),H(x_l),(k=i)_{\Right}(x_l),\\ &
	\pushright{\Child(x,x_r), s_R(x_r),H(x_r),(k=j)_{\Left}(x_r)}\\
	H(x) & \leftarrow \Label_r(x), \Child(x,x_l), s_L(x_l),H(x_l),(k=i)_{\Right}(x_l), \qquad\\ &
	\pushright{\Child(x,x_r), s_R(x_r),H(x_r),(k=j)_{\Left}(x_r).}
\end{align*} Now, a node labeled by $r$ is marked with $H$ if its configuration tree satisfies all horizontal constraints from $H(\Aut{A}_w)$. By the following rules, we ensure that in a configuration tree do not exist two different CCells and use the $\idb$ predicate $\theta_i$ if the CCell $i$ exists in a subtree and $\Non_{\theta}$ if a cell does not belong to $\CCells$. For all $i \in \set{1,\ldots,k}$ where $i \in \CCells$, we add the rule \[\theta_i(x) \leftarrow (k=i)_{\Leaf}(x)\] and for all $j \in \set{1,\ldots,k}$ where $j \notin \CCells$, we add the rules \[\Non_{\theta} (x) \leftarrow (k=j)_{\Leaf}(x)\] to $\PP_1$. This will be propagated by 
\begin{align*}
	\theta_i(x) & \leftarrow  \Child(x,x_l),s_L(x_l), \theta_i(x_l), \Child(x,x_r),s_R(x_r), \Non_{\theta}(x_r) \\
		\theta_i(x) & \leftarrow  \Child(x,x_l),s_L(x_l), \Non_{\theta}(x_l), \Child(x,x_r),s_R(x_r), \theta_i(x_r) \\
\end{align*}
for every $i \in \set{1,\ldots, k}$ where $i\in \CCells$ and finally, a node labeled by $r$ carries 
the  $\idb$ predicate $\theta_i$ for exactly one $i \in \set{1,\ldots,k}$ if its configuration contains exactly one CCell, that is the configuration cell $i$. Otherwise, the node is not marked by any $\theta_i$ predicate. Implied by the following rules 
\[ r(x)  \leftarrow \Label_r(x), H(x), r_{nav}(x), \theta_i(x)  \qquad \text{ for all }i \in \CCells \] every root node of a configuration tree is marked with $r$ if its configuration tree satisfies the conditions (\ref{cond:vGadget})--(\ref{cond:H}).

Purposing the bottom-up analysis of the input tree, we have to verify that a configuration node labeled by $\LeafCTA$, $\LeafCTEleft$, or $\LeafCTEright$ represents a halting configuration that is given as CCell in the leftmost cell of its configuration tree. So, for all $i \in \CCells$ representing a configuration cell with current state $q_a$ that is the only accepting state of $\Aut{A}_w$, we add the rules
\begin{align*}
 \LeafCTA(x) & \leftarrow \Label_{\LeafCTA}(x) , \Child(x,x_r),r(x_r), \theta_i(x_r),(k=i)_{\Left}(x_r) \\
  \LeafCTEleft(x) & \leftarrow \Label_{\LeafCTEleft}(x) , \Child(x,x_r),r(x_r), \theta_i(x_r),(k=i)_{\Left}(x_r) \\
   \LeafCTEright(x) & \leftarrow \Label_{\LeafCTEright}(x) , \Child(x,x_r),r(x_r), \theta_i(x_r),(k=i)_{\Left}(x_r) \\
 \end{align*}
  Recall, the rank of the symbols representing a halting configuration is $\ar(\LeafCTA)=\ar(\LeafCTEright)=\ar(\LeafCTEleft)=1$ and so, for every subtree rooted by a node marked with the latter introduced $\idb$ predicates, we ensured conditions (\ref{cond:CT-r})--(\ref{cond:H}) and (\ref{cond:CTleafAcc}). 
	
	It remains to analyze the subtrees of the $\CT$ labeled nodes. Recall, an inner node of the $\CT$ tree will be positively marked 
	if \begin{enumerate}[(a)]
	\item it is labeled as universal configuration and it has two existential configuration children (one or both can be a leaf configuration node) carrying different CCells, or
	\item it is labeled as existential configuration and it has exactly one universal configuration child (or one leaf configuration node).
	\end{enumerate}
Additionally, it has an $r$ rooted configuration tree as child and the CCell on the $r$ node denotes a state of the machine that is existential if the configuration node is labeled as existential or that is universal if the configuration node is labeled as one.\footnote{Recall, the rank of $\CTEleft$, $\CTEright$, and $\CTA$ is $\ar(\CTEleft)=2$, $\ar(\CTEright)=2$, and $\ar(\CTA)=3$.} So,  we introduce predicates $\state_{\exists}$ and $\state_{\forall}$, as well as we extend the handling of the $\idb$-predicates $\LeafCTA$, $\LeafCTEright$, and $\LeafCTEleft$ by the following rules
\[	\state_{\exists}(x)  \leftarrow r(x), \theta_i(x) \] for all $i \in \CCells$ where $i$ is a configurations cell of an existential state, and \[
	\state_{\forall }(x)  \leftarrow r(x), \theta_j(x) \] for all $j \in \CCells$ where $j$ is a configurations cell of a universal state, and finally, we add \begin{align*}
	  \CTEleft(x) & \leftarrow   \LeafCTEleft(x) \\
  \CTEleft(x) & \leftarrow \state_{\exists}(x),\Label_{\CTEleft}(x) , \Child(x,x_r),r(x_r), \Child(x,x_a), \CTA(x_a),   \\
  	    \CTEright(x) & \leftarrow   \LeafCTEright(x) \\
  \CTEright(x) & \leftarrow  \state_{\exists}(x), \Label_{\CTEright}(x) , \Child(x,x_r),r(x_r), \Child(x,x_a), \CTA(x_a)  \\
  	    	\CTA(x) & \leftarrow 	\LeafCTA(x) \\
	\CTA(x) & \leftarrow \state_{\forall}(x),\Label_{\CTA}(x) , \Child(x,x_r),r(x_r), \\
	&  \pushright{\Child(x,x_1), \CTEleft(x_1),\Child(x_1,x_{1_r}),r(x_{1_r}), \theta_i(x_{1_r}),\qquad \qquad} \\
	& \pushright{\Child(x,x_2), \CTEright(x_2), \Child(x_2,x_{2_r}),r(x_{2_r}), \theta_j(x_{2_r})}
\end{align*} for all $i\neq j \in \{1,\ldots k\}$. Observe, that a node $v$ is marked with $\CTEleft$, $\CTEright$, or $\CTA$ if its subtree rooted by $v$ satisfies the conditions  (\ref{cond:CT-r}) -- (\ref{Cond:diffCCells}) and (\ref{cond:CTleafAcc}). 

Now, to ensure condition (\ref{Cond:startCT}) we fix $i \in \CCells$ that represents the configuration $(\Blank, ((q_0,\Blank),(q_0,\Blank,S)))$ and add the following rules \begin{align*}
 \startCT(x) & \leftarrow \CTA(x), \Child(x,x_r), r(x_r),\theta_i(x_r) \\
  \startCT(x) & \leftarrow \CTEleft(x), \Child(x,x_r), r(x_r),\theta_i(x_r) 
 \end{align*} to $\PP_1$. It is not forbidden that more than one node of the computation tree carries the marker as start configuration node, but the topmost configuration node has to be marked. And therefore, we add the rule \[ structure(x_{\top}) \leftarrow \Label_{\top}(x_{\top}), \Child(x_{\top},x_{\CT}),\startCT(x_{\CT}) \] and obtain a program $\PP_1$ such that a query $Q'=(\PP_1,\structure)$ yields \Yes\ on an input tree $T$ if and only if $T$ satisfies  conditions (\ref{cond:kTOP})--(\ref{cond:CTleafAcc}), that is, if and only if it is structured as an encoded computation tree of $\Aut{A}_w$.

To complete the demanded query $Q=(\PP,\Ans)$, it remains to extend the program $\PP_1$ in a way that $Q$ accepts the tree if the $\structure$ predicate is true for its root and the encoded configurations does not violate the transition relation. For the beginning, let $\PP$ consists of all rules of $\PP_1$. To shorten the query program, we mark all configuration nodes with the predicate $\CT$ by adding the following rules. \begin{align*}
	\CT(x) & \leftarrow \CTA(x)  & 	\LeafCT(x) & \leftarrow \LeafCTA(x)   \\
	\CT(x) & \leftarrow \CTEleft(x) & 	\LeafCT(x) & \leftarrow \LeafCTEleft(x) \\
	\CT(x) & \leftarrow \CTEright(x) & 	\LeafCT(x) & \leftarrow \LeafCTEright(x)  \\
\end{align*}
Since the upcoming rules are very large, we introduce short hands as binary predicates.\footnote{This does not mean that our datalog program is no longer a monadic program, in fact, we use these predicates for replacements in the rule to increase the readability of the whole rule. Variables occurring in the definition of the predicate, but not in the head, have to be renamed in a later context if it is necessary.} First, we define a predicate $\Succ(x_{r_1},x_{r_2})$ that is true for two nodes $x_{r_1}$ and $x_{r_2}$ if they are root nodes of successive encoded configuration trees. \[
\Succ(x_{r_1},x_{r_2}) \deff \left\{  \begin{array}{c} r(x_{r_1}), r(x_{r_2}), \CT(s_1),\CT(s_2),\\ \Child(s_1,s_2), \Child(s_1,x_{r_1}), \Child(s_2,x_{r_2}) \end{array}\right\} 
\]
The next predicate $\SameLevel_i(x_{s_1},x_{s_2})$ for an $i>0$ states for two nodes $x_{s_1}$ and $x_{s_2}$ that they are on the same level $i$ in the configuration tree of two successive encoded configuration trees.
 \[
\SameLevel_i(x_{s_1},x_{s_2}) \deff \left\{ \begin{array}{c}  s(x_{s_1}),s(x_{s_2}),\Succ(x_{r_1},x_{r_2}) ,\\ \Child^i(x_{r_1},x_{s_1}), \Child^i(x_{r_2},x_{s_2}) \end{array} \right\} 
\]
The predicate $\SameLevel^{\textit{LR}}_i(x_{s_1},x_{s_2})$ extends the predicate $\SameLevel_i(x_{s_1},x_{s_2})$ by the following property: The nodes $x_{s_1}$ and $x_{s_2}$ have to be both the left or both the right child of their parent.
 \[
\SameLevel^{\textit{LR}}_i(x_{s_1},x_{s_2}) \deff \left\{ \begin{array}{c} \SameLevel_i(x_{s_1},x_{s_2}),\\ \Child(x_{s_1},x_{p_1}), p(x_{p_1}),\Child(x_{s_2},x_{p_2}),p(x_{p_2}),\\ \Desc(x_{p_1},x_{t_1}),1(x_{t_1}), \Desc(x_{p_2},x_{t_2}) , 1(x_{t_2}) \\ \Child^{i+4}(z,x_{t_1}),\Child^{i+5}(z,x_{t_2}) \end{array} \right\}
\]
Observe, that the node $z$ is the configuration node of the predecessor configuration or its parent node and so, for the initial configuration at the top of the encoded computation tree, the extra buffering node above is necessary. Furthermore, this is the only point during the reduction where the $\Desc$ predicate is actually indispensable; we use it to guess whether the nodes are left or right children. In particular, if the nodes $x_{t_1}$ and $x_{t_2}$ do not indicate the same left- or right-orientation then the distance to $z$ is not $i+4$ for the predecessor and $i+5$ for the successor and a valuation of the rule will not be possible. Even another labeling of the encoding tree that tells us directly whether a child is the left or the right one seems to be impossible because it implies a rule for every path through the configuration tree; that leads to $2^n$ rules and this would avoid a reduction in time polynomial in $n$ and the size of the automaton.

Now, we are able to introduce a predicate  $\SameCell(x_{s_1},x_{s_2})$ that states for two skeleton nodes $x_{s_1}$ and $x_{s_2}$ reflecting the same cell of successive encoded configuration cell sequences; those cells are at depth $n$ of any configuration tree. 
\begin{align*}
 \SameCell(x_{s_1}&,y_{s_2}) \deff \\
  &\bigcup\limits_{1 \leq i \leq n-1} \left\{   \Child(x_i,x_{i+1}),\Child(y_i,y_{i+1}), \SameLevel^{\textit{LR}}_i(x_i,y_i)  \right\} \qquad\\
 & \pushright{\cup \set{\Child(x_{n-1},x_{s_1}),\Child(y_{n-1},y_{s_2}), \SameLevel^{\textit{LR}}_n(x_{s_1},y_{s_2}) } }
 \end{align*}
Next, we use the $\idb$ predicate $\delta$ to denote that a  configuration cell meshes with its  predecessor configuration cell in respect to the transition relation. So, for every tuple
 $(i,j) \in V(\Aut{A}_w)$ we add the following rule \begin{align*}
	\delta(x_{s_2}) &\leftarrow \SameCell(x_{s_1},x_{s_2}), \Child(x_{s_1},x_{m_1}), m(x_{m_1}), m_{k=i}(x_{m_1}),\qquad\\
	&  \pushright{\Child(x_{s_2},x_{m_2}), m(x_{m_2}), m_{k=j}(x_{m_2})}
\end{align*}
to $\PP$. To verify the correctness of this rule, recall that the $m$-labeled node $v$ of an 'me' cell configuration gadget is already marked with $m_{k=i}(v)$ if its gadget encodes the configuration cell $i$. Now, we have to verify that every configuration cell of the encoded sequence  respects the transition relation regarding its predecessor configuration cell and propagate this information to the configuration node by the following rules.
\begin{align*}
	\delta(x)  & \leftarrow \Child(x,x_l), s_L(x_l), \delta(x_l),\Child(x,x_r), s_R(x_r), \delta(x_r)\\
	\delta(x) & \leftarrow \CT(x) ,\Child(x,x_r) , r(x_r), \delta(x_r)
\end{align*}
The next step is to collect the information that every configuration node is a valid successor up to the top of the tree and we obtain that a configuration node $v$ is marked with $\Delta$ if the subtree rootet at $v$ is a suffix of a valid computation tree.
	\begin{align*}
		\Delta(x) & \leftarrow \LeafCT(x) \\
		\Delta(x) & \leftarrow \CTEleft(x) , \Child(x,x_a), \CTA(x_a), \Delta(x_a), \delta(x_a)\\
		\Delta(x) & \leftarrow \CTEright(x) , \Child(x,x_a), \CTA(x_a), \Delta(x_a), \delta(x_a)\\
		\Delta(x) & \leftarrow \CTA(x) , \Child(x,x_1), \CTEleft(x_1), \Delta(x_1), \delta(x_1), \qquad  \quad\\
		& \pushright{\Child(x,x_2), \CTEright(x_2), \Delta(x_2), \delta(x_2)}
	\end{align*}
Clearly, if the topmost configuration tree is an initial configuration and marked with $\Delta$ then we know that the input tree represents a valid accepting computation of $\Aut{A}_w$.  To this end, we conclude the construction by adding the rule \[ Ans(x) \leftarrow structure(x), \Child(x,x_{CT}), \Delta(x_{CT})\] and obtain the demanded query $Q=(\PP,\Ans)$ within polynomial time; that finishes the proof of Theorem~\ref{Thm:Desc_Empti_ranked_ordered}.\end{proofof}

  \section{Hardness on Unranked  Trees}
  
\noindent{\hypersetup{pdfborderstyle={/S/U/W 0}}\textbf{Theorem \ref{Thm:Desc_QCP_unranked_unordered}  \textit{(restated)}} \emph{ The QCP for Boolean $\mDatalog(\tauu^{\Desc})$ on finite labeled unordered trees is $\TwoEXPTIME$-hard.} 
\\ }

\begin{proof}
We prove the theorem by using and extending the proof of Theorem~\ref{Thm:Desc_Empti_ranked_ordered}, so we establish a reduction from the word acceptance problem of exponential space bounded alternating Turing machines to the QCP for  $\mDatalog(\tauu^{\Desc})$  on unranked labeled unordered trees. More precisely, we give a polynomial time reduction to the complement of the named QCP. For a given ATM $\Aut{A}_w$ that is normalized and composed of the original ATM $\Aut{A}$ and its input word $w$, we construct within polynomial time a finite unranked alphabet $\Sigma_{\textit{ur}}$ and two Boolean $\mDatalog(\tau_{u,\Sigma_{ur}}^{\Desc})$-queries $Q_1$ and $Q_2$, such that

\begin{align*}
	w \in L(A) \qquad \iff & \qquad \text{there is an accepting computation tree for ${\Aut{A}_w}$ } \\
						\iff & \qquad \text{there exists an unordered $\Sigma_{\textit{ur}}$-labeled tree $T$ such that }\qquad \\
						& \qquad \qquad   Q_1(T)=\Yes \text{ and } Q_2(T)=\No \\
						\iff & \qquad  Q_1 \not\subseteq Q_2.
\end{align*}
Recall the reduction from Theorem~\ref{Thm:Desc_Empti_ranked_ordered}, the utilized ranked alphabet $\Sigma'$, and the obtained  program $\PP$ in  $\mDatalog(\tauu^{\Desc})$ on ranked trees. We choose the unranked alphabet $\Sigma_{\textit{ur}}$ as the unranked version of $\Sigma'$, to be precise we set $\Sigma_{\textit{ur}}\ \deff \ \set{\alpha | \,\alpha \in \Sigma'}$. Furthermore, we set $Q_1 \deff (\PP,\Ans)$, that is, the query constructed during the former reduction. So, $Q_1$ stands for the ''necessary properties'' of the encoded computation tree. Since the alphabet is no longer ranked, we cannot avoid that a node has more than the planned children, but we can forbid that the redundant children have other labels and falsify the computation. Therefore, all that remains is to construct a query $Q_2$ in $\mDatalog(\tauu^{\Desc})$ such that $Q_2$ describe ''forbidden properties''. A tree with such properties does not describe an encoded computation tree. To this end, we check for forbidden labels on child nodes, a child of an $s$-labeled node, for example, must not be labeled with $\CTA$, and we have to test that there are no two paths encoding inconsistent information. Thus, the query $Q_2=(\PP_2,\Reject)$ will yield to $\Yes$ on an input tree if at least one of the following facts are  true. 

\begin{enumerate}[(1)]
\item \label{cond:OnlyOneRoot} A non root node is labeled by $\top$.
\item \label{cond:RootChildOnlyCT} The root has a child that is not labeled by an $\CT$-label.
\item \label{cond:CteChildOnlyRCta} A non halting existential configuration node has a child labeled with a symbol not in $\{r, \CTA, \LeafCTA \}$.
\item \label{cond:CtaChildOnlyRCte}An non halting universal configuration node has a child labeled with a symbol not in $\{r, \CTEright, \CTEleft, \LeafCTEright, \LeafCTEleft \}$.
\item \label{cond:LeafCtChildOnlyR}A halting  configuration node has a child labeled with a  symbol that is not $r$.
\item \label{cond:RChildOnlyS}An $r$ labeled node has a child labeled with a  symbol that is not $s$.
\item \label{cond:SChildOnlySp}An $s$ labeled node has a child labeled with a symbol not in $\{p,s,s_{\Leaf} \}$.
\item \label{cond:SLeafChildpm}An $s_{\Leaf}$ labeled node has a child labeled with a symbol not in $\{p,m\}$.
\item \label{cond:pmChildOnly01}A $p$ or $m$ labeled node has a child labeled with a symbol not in $\{0,1\}$.
\item \label{cond:01ChildOnly01}A $0$ or $1$ labeled node has a child labeled with a symbol not in $\{0,1,\bot\}$.
\item \label{cond:LeafIsALeaf} A $\bot$ labeled node has a child.
\item \label{cond:pmChildDistBot}A $p$ (or an $m$) labeled node has a descendant that is labeled $\bot$ with distance not equal to three (not equal to $k+1$), or is not a prefix of a valid gadget.
\item \label{cond:height} There exists a path in a configuration tree from the $r$ labeled node to an $s_{\Leaf}$ of length not equal to $n$.
\item  \label{cond:Identic}If any node has two children fulfilling the same role, but  encoding different information.
\end{enumerate}

Obviously, the conditions (\ref{cond:OnlyOneRoot}) -- (\ref{cond:height}) reflect the underlying structure. Additionally, an illustration to condition (\ref{cond:pmChildDistBot}) is given with Figure~\ref{Fig:Extend2Unranked} (a).  Condition (\ref{cond:Identic}) reflects the consistence of the encoding and enforces the following; if there are two configurations as children of a node in the computation tree, both universal, both left -- or right -- existential, then they have to provide exactly the same information during the computation. This includes  the contained  configuration trees, navigation gadgets, and so on, 
which can have different copies or copies of prefixes. Intuitively, it is clear that it does not matter if a node has additional children, but they must not provide wrong information; since every rule uses a maximum distance of $3+n+k$, it suffices  to have a fixed look ahead inside the encoded configuration (cf. Figure~\ref{Fig:Extend2Unranked} (b)). Now, it is comprehensible that the query $Q_2$ fulfilling condition (\ref{cond:OnlyOneRoot})--(\ref{cond:Identic}) yields \No\ on a tree $T$ and $Q_1$ yields \Yes\ on the same tree if and only if $T$ is an encoded  accepting  computation of $\Aut{A}_w$.

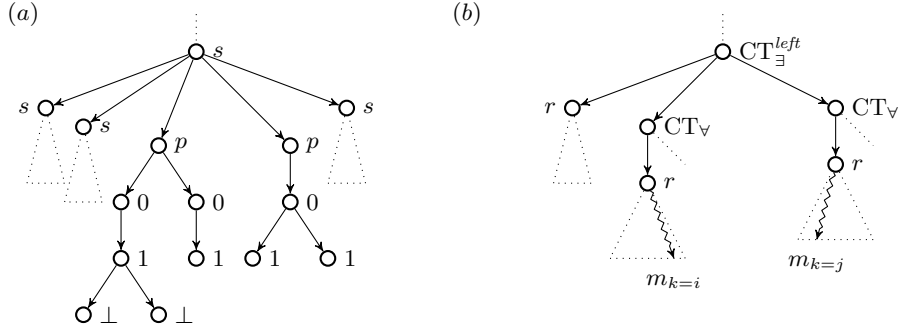
\begin{figure}[h]
\begin{center}
  \begin{tikzpicture}[->,>=stealth']
  	\tikzset{
  treenode/.style = {align=center, inner sep=0pt, text centered,
    font=\sffamily},  	
  arn_red/.style = {treenode, circle, white, font=\sffamily\bfseries, draw=gray,
    fill=white, text width=1em},
  arn_blue/.style = {treenode, circle, black, draw=black,  
    text width=1em, thick},
	}  
  
 \node at (-5.3,0) {$(a)$};  
  \node at (0.6,0) {$(b)$};

    \begin{scope}[xshift=-3cm,scale=0.5]
  	\tikzset{
  treenode/.style = {align=center, inner sep=0pt, text centered,
    font=\sffamily},  	
  arn_red/.style = {treenode, circle, white, font=\sffamily\bfseries, draw=gray,
    fill=white, text width=0.6em},
  arn_blue/.style = {treenode, circle, black, draw=black,  
    text width=0.6em, thick},
	}    
		 \node[arn_blue] at (0,-1) (v1) {};	
    	 \node[anchor=west] at (0.2,-1) (1abel) {$s$};	
    
         \node[arn_blue] at (4,-2.5) (v2) {};	
         \node[anchor=west] at (4.2,-2.5) (1abel) {$s$};		
          \node[arn_blue] at (-3,-3) (v2_1) {};	
         \node[anchor=west] at (-2.8,-3) (1abel) {$s$};		
         \node[arn_blue] at (-4,-2.5) (v2_r) {};	
         \node[anchor=east] at (-4.2,-2.5) (1abel) {$s$};

         \node[arn_blue] at (-1,-3.5) (v6) {};	         
         \node[anchor=west] at (-0.8,-3.5) (1abel) {$p$};		
         \node[arn_blue] at (-2,-5) (v6_0) {};	
         \node[anchor=west] at (-1.8,-5) (1abel) {$0$};
         \node[arn_blue] at (0,-5) (v26_0) {};	
         \node[anchor=west] at (0.2,-5) (1abel) {$0$};
          \node[arn_blue] at (0,-6.5) (v26_1) {};	
         \node[anchor=west] at (0.2,-6.5) (1abel) {$1$};
         
         \node[arn_blue] at (-2,-6.5) (v6_1) {};	
         \node[anchor=west] at (-1.8,-6.5) (1abel) {$1$};
         \node[arn_blue] at (-3,-8) (v6_leaf) {};
         \node[anchor=west] at (-2.8,-8) (1abel) {$\bot$};	
             \node[arn_blue] at (-1,-8) (v26_leaf) {};
         \node[anchor=west] at (-0.8,-8) (1abel) {$\bot$};

         \node[arn_blue] at (2.5,-3.5) (v16) {};	         
         \node[anchor=west] at (2.7,-3.5) (1abel) {$p$};		
         \node[arn_blue] at (2.5,-5) (v16_0) {};	
         \node[anchor=west] at (2.7,-5) (1abel) {$0$};
         \node[arn_blue] at (1.5,-6.5) (v16_1) {};	
         \node[anchor=west] at (1.7,-6.5) (1abel) {$1$};
         \node[arn_blue] at (3.5,-6.5) (v16_leaf) {};	
         \node[anchor=west] at (3.7,-6.5) (1abel) {$1$};

        \path[draw, dotted,-] 	(0,0) edge (v1);

	\path
    	(v1) edge (v2) 
    	    	(v1) edge (v2_r) 
        	    	(v1) edge (v2_1) 
    	    	(v1) edge (v6) 
    	    	(v6) edge (v6_0)
    	    	(v6) edge (v26_0)
    	    	(v26_0) edge (v26_1)
    	    	(v6_0) edge (v6_1)
    	    	(v6_1) edge (v26_leaf)
    	    	    	    	(v6_1) edge (v6_leaf)
    	    	      	    	(v1) edge (v16) 
    	    	(v16) edge (v16_0)
    	    	(v16_0) edge (v16_1)
    	    	    	    	(v16_0) edge (v16_leaf)  ;   	    	
        \path
        	(v2) edge[dotted,-] (3.5,-4.5) 
        	(3.5,-4.5) edge[dotted,-] (4.5,-4.5) 
        	(4.5,-4.5) edge[dotted,-] (v2);
          \path
        	(v2_1) edge[dotted,-] (-2.5,-5) 
        	(-2.5,-5) edge[dotted,-] (-3.5,-5) 
        	(-3.5,-5) edge[dotted,-] (v2_1);  

          \path
        	(v2_r) edge[dotted,-] (-3.5,-4.5) 
        	(-3.5,-4.5) edge[dotted,-] (-4.5,-4.5) 
        	(-4.5,-4.5) edge[dotted,-] (v2_r);

    	\end{scope}      
    	
    \begin{scope}[xshift=4cm,scale=0.5]
  	\tikzset{
  treenode/.style = {align=center, inner sep=0pt, text centered,
    font=\sffamily},  	
  arn_red/.style = {treenode, circle, white, font=\sffamily\bfseries, draw=gray,
    fill=white, text width=0.6em},
  arn_blue/.style = {treenode, circle, black, draw=black,  
    text width=0.6em, thick},
	}    
		 \node[arn_blue] at (0,-1) (v1) {};	
    	 \node[anchor=west] at (0.2,-1) (1abel) {$\CTEleft$};	
    
         \node[arn_blue] at (3,-2.5) (v2) {};	
         \node[anchor=west] at (3.2,-2.5) (1abel) {$\CTA$};		
          \node[arn_blue] at (-2,-3) (v2_l) {};	
         \node[anchor=west] at (-1.8,-3) (1abel) {$\CTA$};		
         \node[arn_blue] at (-4,-2.5) (v2_r) {};	
         \node[anchor=east] at (-4.2,-2.5) (1abel) {$r$};		

         \node[arn_blue] at (3,-4) (r) {};	
         \node[anchor=west] at (3.2,-4) (1abel) {$r$};		
          \node[arn_blue] at (-2,-4.5) (r_l) {};	
         \node[anchor=west] at (-1.8,-4.5) (1abel) {$r$};		         
   		\node[anchor=north] at  (-1.3,-6.7) (1abel) {$m_{k=i}$};		         
   		\node[anchor=north] at  (2.5,-6.2) (1abel) {$m_{k=j}$};		         
         
        \path[draw, dotted,-] 	(0,0) edge (v1);

	\path
    	(v1) edge (v2) 
    	    	(v1) edge (v2_r) 
     	    	(v1) edge (v2_l)
        	    (v2_l) edge (r_l) 
           	    (v2_l) edge[dotted,-] (-1,-4) 
     	    	(v2) edge (r)	
           	    (v2) edge[dotted,-] (4,-3.5);

          \path
        	(v2_r) edge[dotted,-] (-3.5,-4.5) 
        	(-3.5,-4.5) edge[dotted,-] (-4.5,-4.5) 
        	(-4.5,-4.5) edge[dotted,-] (v2_r);  

          \path
        	(r) edge[dotted,-] (4,-6) 
        	(4,-6) edge[dotted,-] (2,-6) 
        	(2,-6) edge[dotted,-] (r); 
        	
        	          \path
        	(r_l) edge[dotted,-] (-3,-6.5) 
        	(-3,-6.5) edge[dotted,-] (-1,-6.5) 
        	(-1,-6.5) edge[dotted,-] (r_l); 

			\path[->, decoration={zigzag,segment length=4,amplitude=.9,
  post=lineto,post length=4pt},font=\scriptsize,
  line join=round]
			(r_l) edge[decorate] (-1.3,-6.5)
			(r) edge[decorate] (2.5,-6);

    	\end{scope}

\end{tikzpicture} 
\end{center}
\caption[next]{(a) An example of allowed ''extentions'' of the encoded computation tree, considered at a navigation gadget that can exist multiple times where a copy also can be reduced to a prefix. (b) If the nodes marked by $m_{k=i}$ and $m_{k=j}$ have the same path through their configuration tree, that is, the same sequence of left and right children, then $i$ must be equal to $j$.   }\label{Fig:Extend2Unranked}
\end{figure}

For the beginning, let $\PP_2$ consist of all rules of $\PP$. We only consider trees $T$ with $Q_1(T)=\Yes$, otherwise we have in any way $Q_1 \subseteq Q_2$, which is enough for the reduction.   To propagate any detected	violation to the root node of the input tree, we propagate the $\Reject$ predicate from any node to the root by adding the following rule \[ \Reject(x) \leftarrow \Child(x,x_1),\Reject(x_1) \]to $\PP_2$. We reflect condition (\ref{cond:OnlyOneRoot}) by adding the rule 
\[ \Reject(x) \leftarrow \Child(x,x_1), \Label_{\top}(x_1). \] Since $Q_1(T)=\Yes$, we know the root is labeled with $\top$ and so, we mirror condition (\ref{cond:RootChildOnlyCT}) by the rule
\[ \Reject(x) \leftarrow \Label_{\top}(x), \Child(x,x_1), \Label_{\alpha}(x_1)\] for every $\alpha \in \Sigma_{\textit{ur}} \setminus \set{\CTA,\CTEleft,\CTEright}$.

To verify condition (\ref{cond:CteChildOnlyRCta}) we add the rules
\begin{align*}
 \Reject(x) &\leftarrow \Label_{\CTEleft}(x), \Child(x,x_1), \Label_{\alpha}(x_1)\\ 
  \Reject(x) &\leftarrow \Label_{\CTEright}(x), \Child(x,x_1), \Label_{\alpha}(x_1)
 \end{align*} for every $\alpha \in \Sigma_{\textit{ur}} \setminus \set{r, \CTA, \LeafCTA }$.

To verify condition (\ref{cond:CtaChildOnlyRCte}) we add the rule
\begin{align*}
 \Reject(x) &\leftarrow \Label_{\CTA}(x), \Child(x,x_1), \Label_{\alpha}(x_1)
 \end{align*}for every $\alpha \in \Sigma_{\textit{ur}} \setminus \set{r, \CTEright, \CTEleft, \LeafCTEright, \LeafCTEleft  }$.

To verify condition (\ref{cond:LeafCtChildOnlyR}) we add the rules
\begin{align*}
 \Reject(x) &\leftarrow \Label_{\LeafCTA}(x), \Child(x,x_1), \Label_{\alpha}(x_1) \\
  \Reject(x) &\leftarrow \Label_{\LeafCTEleft}(x), \Child(x,x_1), \Label_{\alpha}(x_1) \\
   \Reject(x) &\leftarrow \Label_{ \LeafCTEright}(x), \Child(x,x_1), \Label_{\alpha}(x_1) 
 \end{align*} for every $\alpha \in \Sigma_{\textit{ur}} \setminus \set{r}$.
 
To verify condition (\ref{cond:RChildOnlyS}) we add the rule
\begin{align*}
 \Reject(x) &\leftarrow \Label_{r}(x), \Child(x,x_1), \Label_{\alpha}(x_1) 
 \end{align*} for every $\alpha \in \Sigma_{\textit{ur}} \setminus \set{s}$.
 
 To verify condition (\ref{cond:SChildOnlySp}) we add the rule
\begin{align*}
 \Reject(x) &\leftarrow \Label_{s}(x), \Child(x,x_1), \Label_{\alpha}(x_1) 
 \end{align*} for every $\alpha \in \Sigma_{\textit{ur}} \setminus \set{p,s,s_{\Leaf} }$.
 
 To verify condition (\ref{cond:SLeafChildpm}) we add the rule
\begin{align*}
 \Reject(x) &\leftarrow \Label_{s_{\Leaf}}(x), \Child(x,x_1), \Label_{\alpha}(x_1) 
 \end{align*} for every $\alpha \in \Sigma_{\textit{ur}} \setminus \set{p,m }$.
 
 To verify condition (\ref{cond:pmChildOnly01}) we add the rules
\begin{align*}
 \Reject(x) &\leftarrow \Label_{p}(x), \Child(x,x_1), \Label_{\alpha}(x_1) \\
 \Reject(x) &\leftarrow \Label_{m}(x), \Child(x,x_1), \Label_{\alpha}(x_1) 
 \end{align*} for every $\alpha \in \Sigma_{\textit{ur}} \setminus \set{0,1 }$.
 
 To verify condition (\ref{cond:01ChildOnly01}) we add the rules
\begin{align*}
 \Reject(x) &\leftarrow \Label_{0}(x), \Child(x,x_1), \Label_{\alpha}(x_1) \\
 \Reject(x) &\leftarrow \Label_{1}(x), \Child(x,x_1), \Label_{\alpha}(x_1) 
 \end{align*} for every $\alpha \in \Sigma_{\textit{ur}} \setminus \set{0,1,\bot }$. 
 
  To verify condition (\ref{cond:LeafIsALeaf}) we add the rule
\begin{align*}
 \Reject(x) &\leftarrow \Label_{\bot}(x), \Child(x,x_1)
 \end{align*} to $\PP_2$. 

To verifiy condition (\ref{cond:pmChildDistBot}), we assume that conditions (\ref{cond:pmChildOnly01})  and (\ref{cond:01ChildOnly01})  are not fulfilled. This implies, the only possible labels at nodes descending a node labeled with $p$ or $m$ are $0$, $1$, and $\bot$. So, we add the following rules that ensures that no $\Leaf$ labeled path exists that is too short or too long, that is a path with a node labeled with $0$ or $1$ on position three for a navigation gadget and on position $k+1$ for a 'me' cell gadget. Thus, we add for the navigation gadget the following rules
\begin{align*}
	\Reject(x) &\leftarrow \Label_{p}(x), \Child(x,x_1), \Label_{\bot}(x_1)\\
	\Reject(x) &\leftarrow \Label_{p}(x), \Child^2(x,x_1), \Label_{\bot}(x_1)\\
	\Reject(x) &\leftarrow \Label_{p}(x), \Child^3(x,x_1), \Label_{0}(x_1)\\
		\Reject(x) &\leftarrow \Label_{p}(x), \Child^3(x,x_1), \Label_{1}(x_1)
\end{align*} and for the 'me' gadget, we add
\begin{align*}
	\Reject(x) &\leftarrow \Label_{m}(x), \Child(x,x_1), \Label_{\bot}(x_1)\\
	\Reject(x) &\leftarrow \Label_{m}(x), \Child^2(x,x_1), \Label_{\bot}(x_1)\\
		\vdots &\\
		\Reject(x) &\leftarrow \Label_{m}(x), \Child^k(x,x_1), \Label_{\bot}(x_1)\\
	\Reject(x) &\leftarrow \Label_{m}(x), \Child^{k+1}(x,x_1), \Label_{0}(x_1)\\
		\Reject(x) &\leftarrow \Label_{m}(x), \Child^{k+1}(x,x_1), \Label_{1}(x_1).
\end{align*} Recall, $\Child^i(x,y)$ is a short hand for the set of atoms denoting $y$ as a descendant of $x$ in the $i$-th generation.

	By the same way, we reflect condition (\ref{cond:height}) which says that  there exists a path in a configuration tree from the $r$ labeled node to an $s_{\Leaf}$ of length not equal to $n$. We know by conditions (\ref{cond:RChildOnlyS}), (\ref{cond:SChildOnlySp}), (\ref{cond:pmChildOnly01}) -- (\ref{cond:LeafIsALeaf})  that it suffices to test if there is a shorter path ending on an $s_{\Leaf}$ labeled node, or if there exists a path of length $n$ ending with an $s$ labeled node. Therefore, we add the rules
	\begin{align*}
		\Reject(x) &\leftarrow \Label_{r}(x), \Child(x,x_1), \Label_{s_{\Leaf}}(x_1)\\
	\Reject(x) &\leftarrow \Label_{r}(x), \Child^2(x,x_1), \Label_{s_{\Leaf}}(x_1)\\
		\vdots &\\
		\Reject(x) &\leftarrow \Label_{r}(x), \Child^{n-1}(x,x_1), \Label_{s_{\Leaf}}(x_1)\\
	\Reject(x) &\leftarrow \Label_{r}(x), \Child^{n}(x,x_1), \Label_{s}(x_1)
	\end{align*}
 to $\PP_2$.

Finally, we consider condition (\ref{cond:Identic}) and we start by verifying all neighboring navigation gadgets. By condition (\ref{cond:pmChildDistBot}) we already know every navigation gadget is a  valid navigation gadget or the prefix thereof . By the following rules, we detect if they are in conflict.
	\begin{align*}
		\Reject(x) &\leftarrow \Child(x_s,x_{p_1}),\Child(x_s,x_{p_2}), \Label_{p}(x_{p_1}),\Label_{p}(x_{p_2}),\qquad \quad\\
		&\pushright{ \Child(x_{p_1},x_1), \Label_{1}(x_1),\Child(x_{p_2},x_0), \Label_{0}(x_0)}\\
		\Reject(x) &\leftarrow \Child(x_s,x_{p_1}),\Child(x_s,x_{p_2}), \Label_{p}(x_{p_1}),\Label_{p}(x_{p_2}),\\
		&\pushright{ \Child^2(x_{p_1},x_1), \Label_{1}(x_1),\Child^2(x_{p_2},x_0), \Label_{0}(x_0)}
	\end{align*} The same holds for the 'me' cell configuration gadget and therefore, we add the rules
	\begin{align*}
		\Reject(x) &\leftarrow \Child(x_s,x_{m_1}),\Child(x_s,x_{m_2}), \Label_{m}(x_{m_1}),\Label_{m}(x_{m_2}),\qquad \quad\\
		&\pushright{ \Child(x_{m_1},x_1), \Label_{1}(x_1),\Child(x_{m_2},x_0), \Label_{0}(x_0)}\\
		\Reject(x) &\leftarrow \Child(x_s,x_{m_1}),\Child(x_s,x_{m_2}), \Label_{m}(x_{m_1}),\Label_{m}(x_{m_2}),\qquad \quad\\
		&\pushright{ \Child^2(x_{m_1},x_1), \Label_{1}(x_1),\Child^2(x_{m_2},x_0), \Label_{0}(x_0)}\\
		\vdots &\\
	\Reject(x) &\leftarrow \Child(x_s,x_{m_1}),\Child(x_s,x_{m_2}), \Label_{m}(x_{m_1}),\Label_{m}(x_{m_2}),\qquad \quad\\
		&\pushright{ \Child^k(x_{m_1},x_1), \Label_{1}(x_1),\Child^k(x_{m_2},x_0), \Label_{0}(x_0)}\\
	\end{align*}
Now, we are going to compare the configurations; that will be done analogously to the definition of the short hand predicate $\SameCell$ in the previous proof, but without the offset that was used to reach the successor configuration. So, we first define the predicates $\EquiLevel$, $\EquiLevel^{\textit{LR}}$, and $\EquiCell$, stating that two nodes are in the equivalent level, are both a left or both a right child, and, by the latter, denote equivalent cells.

 \[
\EquiLevel_i(x_{s_1},x_{s_2}) \deff \left\{ \begin{array}{c}  \Child^2(x,x_{r_1}),\Child^2(x,x_{r_2}),r(x_{r_1}),r(x_{r_2}) , \\  \Child^i(x_{r_1},x_{s_1}), \Child^i(x_{r_2},x_{s_2}),s(x_{s_1}),s(x_{s_2}) \end{array} \right\}
\]
The predicate $\EquiLevel^{\textit{LR}}_i(x_{s_1},x_{s_2})$ extends the predicate $\EquiLevel_i(x_{s_1},x_{s_2})$ by the following property: The nodes $x_{s_1}$ and $x_{s_2}$ have to be both the left or both the right child of their parent.
 \[
\EquiLevel^{\textit{LR}}_i(x_{s_1},x_{s_2}) \deff \left\{ \begin{array}{c} \EquiLevel_i(x_{s_1},x_{s_2}),\\ \Child(x_{s_1},x_{p_1}), p(x_{p_1}),\Child(x_{s_2},x_{p_2}),p(x_{p_2}),\\ \Desc(x_{p_1},x_{t_1}),1(x_{t_1}), \Desc(x_{p_2},x_{t_2}) , 1(x_{t_2}) \\ \Child^{i+4}(z,x_{t_1}),\Child^{i+4}(z,x_{t_2}) \end{array} \right\}
\]
And finally, we define $\EquiCell$ that is true for two nodes denoting  configuration cells that encode the same cell of the automaton. Note, that the predicate is reflexive.
\begin{align*}
 \EquiCell(x_{s_1},y_{s_2}) &\deff \bigcup\limits_{1 \leq i \leq n-1} \left\{   \Child(x_i,x_{i+1}),\Child(y_i,y_{i+1}), \EquiLevel^{\textit{LR}}_i(x_i,y_i)  \right\} \qquad\\
 & \pushright{\cup \set{\Child(x_{n-1},x_{s_1}),\Child(y_{n-1},y_{s_2}), \EquiLevel^{\textit{LR}}_n(x_{s_1},y_{s_2}) } }
 \end{align*}
To verify the value $k$, we utilize the predicate $m_{k=i}$ for  every $i \in \{1,\ldots, k\}$ given by a positive evaluation of query $Q_1$, and compare them for every $i,j  \in \{1,\ldots, k\}$ with $i\neq j$ by the following rules
\begin{align*}
	\Reject(x) \leftarrow &\type(x_{\CT_1}), \type(x_{\CT_2}),\Child^{n+1}(x_{\CT_1},x_{s_1}), \Child^{n+1}(x_{\CT_2},x_{s_2}), \\ &\EquiCell(x_{s_1},x_{s_2}), \Child(x_{s_1},x_{m_1}), m(x_{m_1}), m_{k=i}(x_{m_1}),\qquad\\
	&  \pushright{\Child(x_{s_2},x_{m_2}), m(x_{m_2}), m_{k=j}(x_{m_2})}
\end{align*} for every  $\type \in \set{\CTA,\LeafCTA,\CTEleft,\LeafCTEleft,\CTEright,\LeafCTEright}$.

Now, it is ensured that two configurations in the same role, provide different information, so the demanded query is defined by $Q_2=(\PP_2,\Reject)$.
	
Observe, by $Q_1$ we evaluate the computation tree by starting in the halting configurations, so it does not matter if a configuration has a successor configuration twice or if these successor configurations themselves have different successor configurations. In this case it suffices if one subtree leads to  accepting configurations on the leaves of an appropriate subtree. \end{proof}

\end{document}